\documentclass[11pt,a4paper]{article}
\makeatletter
\RequirePackage[utf8]{inputenc}
\RequirePackage[T1]{fontenc}
\RequirePackage[english]{babel}
\RequirePackage{lmodern,microtype}
\RequirePackage[a4paper,top=2.5cm,bottom=2.5cm,left=2.7cm,right=2.7cm,footskip=1.2cm]{geometry}
\RequirePackage{amsmath,amssymb,amsthm,mathtools,bm}
\RequirePackage{graphicx,booktabs,array,tabularx,multirow,makecell,longtable}
\RequirePackage{caption,subcaption,xcolor,enumitem,csquotes,tikz,siunitx}
\RequirePackage{needspace}
\usetikzlibrary{arrows.meta,positioning,shapes.geometric}
\RequirePackage[backend=biber,style=authoryear-comp,natbib=true,url=true,doi=true,isbn=false,maxbibnames=99,maxcitenames=2,uniquelist=false,sorting=nyt]{biblatex}
\RequirePackage[colorlinks=true,linkcolor=black,citecolor=black,urlcolor=black]{hyperref}
\RequirePackage[noabbrev,capitalize]{cleveref}
\newcolumntype{L}[1]{>{\raggedright\arraybackslash}p{#1}}
\newcolumntype{Y}{>{\raggedright\arraybackslash}X}
\setlist{nosep,leftmargin=*}
\let\ffraStandardSection\section
\renewcommand{\section}{\Needspace{7\baselineskip}\ffraStandardSection}
\theoremstyle{plain}
\newtheorem{proposition}{Proposition}

\newtheorem{lemma}[proposition]{Lemma}

\theoremstyle{definition}
\newtheorem{definition}{Definition}

\theoremstyle{remark}
\newcommand{\ffraseries}{ForesightFlow Research}
\newcommand{\ffraauthor}{Maksym Nechepurenko}
\newcommand{\ffraaffiliation}{Research Department of Devnull FZCO, Dubai, UAE}
\newcommand{\ffraemail}{maksym@devnull.ae}
\newcommand{\ffraorcid}{0000-0002-9515-8841}
\newcommand{\ffraversion}{r0.0.0}
\newcommand{\ffradate}{\today}

\newcommand{\ffraSetup}[1]{\setkeys{ffra}{#1}}
\RequirePackage{keyval}
\define@key{ffra}{series}{\renewcommand{\ffraseries}{#1}}
\define@key{ffra}{author}{\renewcommand{\ffraauthor}{#1}}
\define@key{ffra}{affiliation}{\renewcommand{\ffraaffiliation}{#1}}
\define@key{ffra}{email}{\renewcommand{\ffraemail}{#1}}
\define@key{ffra}{orcid}{\renewcommand{\ffraorcid}{#1}}
\define@key{ffra}{version}{\renewcommand{\ffraversion}{#1}}
\define@key{ffra}{date}{\renewcommand{\ffradate}{#1}}
\define@key{ffra}{status}{\renewcommand{\ffrastatus}{#1}}
\renewcommand{\maketitle}{\begin{center}\small\textsc{\ffraseries}\par\vspace{0.75em}{\LARGE\bfseries \@title\par}\vspace{0.8em}{\large \ffraauthor\par}\vspace{0.25em}{\small \ffraaffiliation\par\texttt{\ffraemail}\quad ORCID \ffraorcid\par\vspace{0.25em}\ffradate\quad\textbar\quad Version \ffraversion\par}\end{center}\vspace{0.5em}}
\newcommand{\ffraKeywords}[1]{\noindent\textbf{Keywords:} #1\par}
\newcommand{\ffraJEL}[1]{\noindent\textbf{JEL:} #1\par}
\newcommand{\ffraEvidence}[1]{\PackageWarning{foresightflow_research_article}{\string\ffraEvidence is deprecated and renders no text in FFRA-v1.0.3}}

\definecolor{seriesblue}{RGB}{31,78,121}
\definecolor{seriesgray}{RGB}{242,244,247}
\definecolor{serieslight}{RGB}{248,249,251}
\definecolor{seriesdark}{RGB}{52,58,64}
\definecolor{seriesorange}{RGB}{180,96,33}
\usetikzlibrary{fit,calc,decorations.pathreplacing,matrix,backgrounds,patterns}
\tikzset{seriesbox/.style={draw=black!70,rounded corners=2pt,align=center,inner sep=5pt,minimum height=9mm,fill=seriesgray},seriesarrow/.style={-{Latex[length=2.2mm]},thick,draw=black!75}}
\makeatother
\newcommand{\ObsExact}{\mathsf{E}}
\newcommand{\ObsInterval}{\mathsf{I}}
\newcommand{\ObsSnapshot}{\mathsf{S}}
\newcommand{\ObsProxy}{\mathsf{Prx}}
\newcommand{\ObsUnmeasured}{\mathsf{U}}
\newcommand{\ObsConflict}{\mathsf{X}}

\newtheorem{empiricalfinding}{Empirical Finding}
\newtheorem{designprinciple}{Design Principle}
\crefname{designprinciple}{Design Principle}{Design Principles}
\Crefname{designprinciple}{Design Principle}{Design Principles}
\crefname{empiricalfinding}{Empirical Finding}{Empirical Findings}
\Crefname{empiricalfinding}{Empirical Finding}{Empirical Findings}
\usepackage{pgfplots}
\pgfplotsset{compat=1.18}
\usepackage{adjustbox}
\usepackage{etoolbox}
\hypersetup{
  pdftitle={Resolution Is Not Settlement, Part I: Oracle Adjudication and Semantic Governance on Polymarket},
  pdfauthor={Maksym Nechepurenko},
  pdfsubject={Prediction markets; oracle finality; semantic governance; Polymarket},
  pdfkeywords={prediction markets, resolution, settlement, finality, oracle, Polymarket, leverage}
}

\newcommand{\tevent}{t_{\mathrm{event}}}
\newcommand{\tsource}{t_{\mathrm{source}}}
\newcommand{\teligible}{t_{\mathrm{eligible}}}
\newcommand{\trequest}{t_{\mathrm{request}}}
\newcommand{\tproposal}{t_{\mathrm{proposal}}}
\newcommand{\tdispute}{t_{\mathrm{dispute}}}
\newcommand{\treset}{t_{\mathrm{reset}}}
\newcommand{\toracle}{t_{\mathrm{oracle\text{-}final}}}
\newcommand{\tadapter}{t_{\mathrm{adapter\text{-}terminal}}}
\newcommand{\tclarify}{t_{\mathrm{clarification}}}
\newcommand{\sourceinterval}{\mathcal I_q^{S}}
\newcommand{\eligibleinterval}{\mathcal I_q^{E}}
\newcommand{\readyinterval}{\mathcal I_q^{\dagger}}

\newcommand{\Prob}{\mathbb P}
\newcommand{\ind}[1]{\mathbf 1\!\left[#1\right]}
\newcommand{\occupancy}{\Omega}

\newcommand{\code}[1]{\texttt{#1}}

\crefname{designprinciple}{Design Principle}{Design Principles}
\Crefname{designprinciple}{Design Principle}{Design Principles}
\crefname{empiricalfinding}{Empirical Finding}{Empirical Findings}
\Crefname{empiricalfinding}{Empirical Finding}{Empirical Findings}

\newcommand{\tproposalany}{t_{\mathrm{proposal,any}}}
\newcommand{\tproposalterminal}{t_{\mathrm{proposal,terminal}}}
\newcommand{\trequestcreated}{t_{\mathrm{request-created}}}
\newcommand{\tconsume}{t_{\mathrm{consume}}}
\newcommand{\tprotocol}{t_{\mathrm{protocol\text{-}final}}}

\newcommand{\route}{G_q}
\newcommand{\pathflags}{\mathbf Z_q}

\ffraSetup{series={ForesightFlow Research \textperiodcentered{} Event-Linked Perpetuals \textperiodcentered{} Paper 5, Part I},author={Maksym Nechepurenko},affiliation={Research Department of Devnull FZCO, Dubai, UAE},email={maksym@devnull.ae},orcid={0000-0002-9515-8841},version={r0.6.3},date={August 2026},status={}}

\title{Resolution Is Not Settlement, Part I:\\
Oracle Adjudication and Semantic Governance on Polymarket}
\author{\SeriesAuthor\thanks{\SeriesAffiliation}}
\date{\SeriesDate}

\begin{document}
\maketitle
\begin{abstract}
Prediction-market resolution is often reduced to a terminal outcome and one timestamp. That representation is inadequate for leveraged event claims because rule versioning, technical request creation, proposal, dispute, reset, Oracle finality, and adapter terminality are distinct states with different observation precision and balance-sheet consequences. We reconstruct those states for Polymarket using Oracle request generations, rather than economic questions alone, as the unit of adjudication.

The population is frozen at Polygon block 79,721,080 and contains 185,550 initialized adapter-question instances and 350,703 decoded adapter logs. Exact requester-filtered extraction yields 504,332 decoded Oracle lifecycle events comprising 184,148 request creations, 159,447 proposals, 1,604 disputes, and 159,133 settlements. The request-generation accounting closes as 182,671 questions with at least one request plus 1,477 successor generations. Immutable chain identity and deployed request semantics provide exact linkage; no nearest-time, pooled-question, or question-ID-only fallback is used. Unfinished histories remain right-censored.

Request age is not semantic resolution age. Median request-to-first-proposal time is 182 seconds on the legacy route but 176,388--744,151 seconds on modern routes, while post-reset successor proposals arrive within 300--2,909 seconds at the median. The contrast is descriptive of different mechanism start points and is not, by itself, evidence of causal efficiency. Exact stable-ID metadata linkage recovers 104,032 of 185,550 questions (56.07\%), leaves 81,518 unmatched, and produces no ambiguous exact match. The versioned rule layer contains 1,570 updates; exact chain identity timestamps 1,561 of them, including all 1,549 creator-authoritative clarifications. Among 1,711 exact clarification--generation relations, 1,426 occur after request creation but before first proposal. External-source publication and contractual-decidability clocks remain unmeasured population-wide and are not replaced by mechanism timestamps.

The results establish an event-sourced account of Oracle adjudication, semantic governance, and adapter terminality. A separately frozen companion Part~II reconstructs Conditional Tokens payout recording and observed redemption, preserving rather than collapsing the boundary between adjudication, protocol settlement, and holder realization.
\end{abstract}

\ffraKeywords{prediction markets; event-linked perpetual futures; finality; Polymarket.}
\ffraJEL{G13, G14, G18.}

\section{Introduction}
\label{sec:introduction}

An event can be economically over while its market remains institutionally unfinished.  A
match may have ended, an official result may be available, and the eventual payout may appear
obvious to almost every trader.  Yet the corresponding claim can remain inside a mechanism
that still has to identify the governing rule version, accept a proposal, permit a challenge,
process a reset or later adjudication path, make the successful oracle value consumable, and
record a terminal adapter transition.  Compressing those steps into one field such as
\code{resolved\_at} is relatively harmless when the only object of interest is the final answer.
It is not harmless when capital is margined, funding accrues, a liquidation trigger can transfer
value, or governance concentration and rule changes are themselves objects of study.

The distinction in the title is literal.  In this study, \emph{resolution} is the
adjudication process that selects a contractually accepted outcome.  \emph{Settlement} is a
downstream protocol action that makes a payout enforceable against locked collateral.
\emph{Redemption} is a holder-level action that realizes collateral.  These layers may share a
transaction or a block timestamp, but they are different economic states.  Part I reconstructs
the Oracle and adapter layers.  The companion Part II follows the same research programme into
the Conditional Tokens Framework (CTF), payout recording, redeemability, and observed holder
realization under a separately frozen downstream dataset.

The product boundary matters.  Polymarket documents prediction markets and a separate
perpetual-futures product surface.  The latter tracks continuous external underlyings and does
not resolve to a binary event payoff \citep{polymarket_perps_faq_2026,polymarket_perps_markets_2026}.
It is therefore not a deployed test of an event-probability perpetual.  This paper studies the
existing event-resolution stack that any future event-linked leveraged instrument would inherit.

\subsection{Why one resolution timestamp is not enough}

A question can be initialized long before it is contractually decidable.  A proposal can carry an
ignore or too-early value rather than a terminal answer.  A first dispute can reset the question
and create a successor request.  A later dispute can enter a different adjudication branch.  A
creator-authoritative clarification can arrive after technical request creation but before the
first proposal.  Oracle settlement can occur before or inside the transaction in which the
adapter reaches its terminal state.  These are not cosmetic distinctions: they alter the
available next transitions, the time spent in costly states, the identity of actors who can
transfer value, and the evidence that a risk engine could observe in real time.

\begin{figure}[htbp]
\centering
\resizebox{\textwidth}{!}{%
\begin{tikzpicture}[
  node distance=5mm and 7mm,
  every node/.style={font=\footnotesize},
  layer/.style={seriesbox,minimum width=2.45cm,minimum height=1.05cm},
  outside/.style={seriesbox,dashed,minimum width=2.45cm,minimum height=1.05cm,fill=white},
  arr/.style={seriesarrow}
]
\node[outside] (evidence) {External evidence\\and rule eligibility};
\node[layer,right=of evidence] (request) {Technical request\\generation};
\node[layer,right=of request] (proposal) {Proposal, liveness,\\dispute, reset};
\node[layer,right=of proposal] (oracle) {Oracle finality};
\node[layer,right=of oracle] (adapter) {Adapter terminal\\transition};
\node[outside,right=of adapter] (protocol) {Protocol payout\\recording};
\node[outside,right=of protocol] (redeem) {Holder redemption\\and cash realization};
\draw[arr] (evidence) -- (request);
\draw[arr] (request) -- (proposal);
\draw[arr] (proposal) -- (oracle);
\draw[arr] (oracle) -- (adapter);
\draw[arr] (adapter) -- (protocol);
\draw[arr] (protocol) -- (redeem);
\draw[very thick,dashed,seriesblue] ($(adapter.north east)+(3mm,7mm)$) -- ($(adapter.south east)+(3mm,-7mm)$);
\node[font=\scriptsize\bfseries,seriesblue,align=center] at ($(adapter.east)+(3mm,13mm)$) {Part I\\boundary};
\node[font=\scriptsize,align=center] at ($(request.south)!0.5!(adapter.south)+(0,-10mm)$) {Observed on-chain adjudication and adapter consumption};
\node[font=\scriptsize,align=center] at ($(protocol.south)!0.5!(redeem.south)+(0,-10mm)$) {Companion Part II};
\end{tikzpicture}%
}
\caption{Finality layers and the analytical scope boundary between the two companion papers.  Part I reconstructs request-generation adjudication through adapter consumption.  External evidence and contractual decidability are separate semantic clocks; protocol payout recording and holder cash realization are downstream layers studied in Part II.}
\label{fig:finality_layers}
\end{figure}
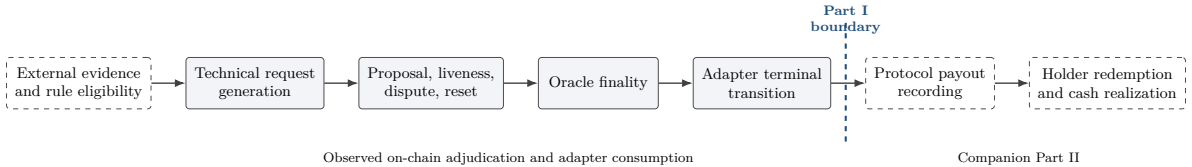

The boundary in \Cref{fig:finality_layers} is analytical rather than an assertion that every
downstream event occurs in a later block.  Part I uses the adapter-question and exact Oracle
request generation as its primary units.  Part II uses the condition, payout vector, and
redemption event.  The two analyses use distinct denominators and do not pool them.  A transaction-atomic
implementation can produce zero observed wall-clock duration between two logical layers
without making them the same state.

\subsection{Position in the research programme}

The revised first four papers establish the surrounding problem.  Paper~1 proves that ordinary
cryptocurrency-perpetual mechanics are not portable to a financed binary claim held through
an adverse terminal outcome under its stated balance-sheet assumptions.  Paper~2 shows that
risk depends on underlying geometry, temporal structure, settlement structure, and venue
composition.  Paper~3 separates market-price, outcome, and resolution manipulation from
informed trading.  Paper~4 documents extreme fill-side concentration while proving that public
fills do not identify the address-level quote lifecycle
\citep{nechepurenko2026_p1,nechepurenko2026_p2,nechepurenko2026_p3,nechepurenko2026_p4}.
The present paper moves below those overlay designs and reconstructs the adjudication states
they would inherit.

Two later papers use this reconstruction rather than duplicating it.  Paper~6 studies
state-triggered collateral, funding, closing auctions, and conversion into funded outcome
claims.  Paper~9 studies multi-leg instruments when two event legs reach finality at different
times.  The present paper supplies their state variables and empirical clocks, not their
mechanism designs.

\subsection{Research questions}

The paper asks six questions.

\begin{enumerate}
  \item What is the complete frozen population of initialized adapter questions, and how many
        exact oracle request generations are associated with it?
  \item Which proposal, dispute, reset, successor-request, settlement, and adapter-terminal
        paths are observed, and which questions or generations remain right-censored?
  \item Which mechanism intervals are observed exactly, which are interval-observed, and
        which semantic clocks remain unavailable?
  \item How concentrated are proposal and dispute execution addresses, subject to the
        boundary that an address is not a beneficial owner?
  \item How are rules and creator-authoritative clarifications versioned, and where do those
        updates fall relative to request, proposal, dispute, reset, and settlement?
  \item Which claims about leveraged controls follow from the observed state surface, and
        which remain questions for Papers~6 and~9?
\end{enumerate}

\subsection{Main empirical findings}

\begin{empiricalfinding}[A question is not a request generation]
\label{ef:a-question-is-not-a-request-generation}
The cutoff population contains 185,550 initialized adapter-questions and 184,148 exact oracle
request generations.  Reset paths create successor generations without changing the economic
question.  Question-level state alone is therefore insufficient for path reconstruction.
\end{empiricalfinding}

\begin{empiricalfinding}[Technical request age is not semantic resolution age]
\label{ef:technical-request-age-is-not-semantic-resolution-age}
Modern request-to-first-proposal medians range from 176,388 to 744,151 seconds, whereas
reset-to-successor-proposal medians range from 300 to 2,909 seconds.  The first interval can
begin before evidence and contractual decidability; the second begins inside an already active
contested path.  They do not share one economic interpretation.
\end{empiricalfinding}

\begin{empiricalfinding}[Rule state can change after the request already exists]
\label{ef:rule-state-can-change-after-the-request-already-exists}
The versioned rule ledger contains 1,570 observed updates.  Exact chain identity timestamps
1,561 updates, including 1,549 creator-authoritative clarifications.  Among 1,711 exact
clarification--generation relations, 83.34\% occur after request creation but before the first
proposal.
\end{empiricalfinding}

\begin{empiricalfinding}[Exact metadata linkage is substantial but selective]
\label{ef:exact-metadata-linkage-is-substantial-but-selective}
Frozen nested Gamma records produce exact stable-ID matches for 104,032 of 185,550
adapter-questions (56.07\%).  A further 81,518 remain unmatched and none are ambiguous under
the exact-ID rule.  Metadata-conditioned estimates therefore apply to the matched cohort and
retain route-specific attrition.
\end{empiricalfinding}

\begin{empiricalfinding}[The public mechanism reveals some clocks and not others]
\label{ef:the-public-mechanism-reveals-some-clocks-and-not-others}
The data measure request, proposal, dispute, reset, settlement, adapter, and clarification
clocks.  They do not measure population-wide external-source publication or contractual
decidability.  Request-to-proposal duration is therefore a mechanism interval, not a
truth-to-resolution interval.
\end{empiricalfinding}

\subsection{Contributions}

The paper makes seven contributions.

First, it establishes a manifest-bound adapter denominator and exact request-generation
linkage without temporal or pooled-question fallback.  Second, it reconstructs a two-level
multistate lifecycle with right censoring, transient resets, and route-specific competing paths.
Third, it distinguishes exact, block-bounded, interval-observed, and unavailable clocks rather
than forcing every duration into a point estimate.  Fourth, it recovers an exact stable-ID
metadata relation with explicit route-dependent attrition and no ambiguous match.  Fifth, it
restores a versioned rule and clarification layer, showing that technical request creation can
precede the on-chain rule update that shapes later adjudication.  Sixth, it documents
proposal/dispute execution concentration while preserving the distinction between addresses,
access configuration, and beneficial ownership.  Seventh, it translates the state
reconstruction into constraints for later leveraged designs while keeping protocol payout and
redemption estimands in the separately identified companion Part II.

\subsection{Evidence discipline}

The paper separates documentary description, source-code capability, deployed configuration,
and observed historical use.  A function in source does not prove invocation.  A current
whitelist does not prove historical access.  A current metadata page is not projected backward.
An execution address is not treated as a beneficial owner.  A request timestamp is not
substituted for source publication or contractual decidability.

All percentages state their denominator.  Questions unresolved at the cutoff are right-censored
rather than classified as permanent failures.  Later-dispute paths are reported with interval or
upper-bound language when exact Data Verification Mechanism (DVM) result time is
unavailable.  The phrase ``version-locked'' is used for hypotheses and estimators frozen before the
corresponding accepted analysis was executed; the paper does not claim an external
preregistration where none occurred.

\subsection{Roadmap}

\Cref{sec:related} positions the contribution.  \Cref{sec:institution} reconstructs the native
architecture.  \Cref{sec:formal} defines the entity hierarchy, clocks, state process, and
identification results.  \Cref{sec:data} describes the frozen data layers and validity gates.
\Cref{sec:estimands} fixes the empirical design.  \Cref{sec:results_adapter,sec:results_oracle,sec:results_semantic}
report the adapter, oracle, and semantic-governance results.
\Cref{sec:design_implications} states consequences for leveraged design without duplicating
Paper~6.  Integrated appendices contain the operational dictionary, estimand and gate registry,
source audit, proofs, hypothesis disposition, and reproducibility and evidence-lineage material.  A concise
data-and-code availability statement appears after the conclusion.

\section{Related Work}
\label{sec:related}

The paper connects five literatures: prediction-market rules, oracle-mediated adjudication,
layered financial finality, multistate duration analysis, and event-linked leverage.
\subsection{Prediction markets and resolution rules}

The prediction-market literature has primarily studied information aggregation, calibration, and the interpretation of market prices \citep{wolfers_zitzewitz_2004,manski_2006}.  Combinatorial market design expands the state space of claims \citep{hanson_2003}, but empirical work normally reduces settlement to the realized payoff.  For the present paper, the payoff is not a sufficient statistic.  Two markets with the same terminal payout can differ in source-publication delay, proposal latency, dispute/reset path, clarification history, protocol-settlement time, and holder-redemption behavior.

Polymarket's CTF-based markets make several intermediate states observable.  The adapter stores ancillary data, prepares a condition, requests a YES/NO oracle value, and later maps a valid value into a payout vector.  Ordinary terminal values include YES, NO, and UNKNOWN/50--50; an ignore or too-early sentinel causes continuation/reset rather than terminal payout \citep{uma_polymarket_verification_2026,polymarket_uma_adapter_contract_2026}.  This contract structure motivates a process representation rather than a single resolution label.

The resolution rules themselves also have a temporal dimension.  A title is not the governing rule; named sources, end dates, edge cases, and later clarifications determine when and how a terminal answer is available \citep{polymarket_resolution_docs_2026}.  The empirical contribution of Paper~5 is to treat those rule states as versioned inputs rather than as static metadata attached after the fact.

\subsection{The oracle problem, access, and dispute latency}

Blockchains cannot observe off-chain facts without an oracle layer.  The oracle literature emphasizes that data-source selection, feeder access, aggregation, dispute, and final on-chain consumption are separate design modules with distinct trust assumptions and latency costs \citep{eskandari_oracles_2021}.  Binary event outcomes add semantic heterogeneity: the problem is not only whether a feed is numerically accurate, but whether the question, source, and edge-case rule identify the same proposition at each stage.

UMA's Optimistic Oracle accepts a bonded claim if it survives liveness and uses the Data Verification Mechanism (DVM) as a disputed-data backstop \citep{uma_oracle_overview_2026}.  ManagedOptimisticOracleV2 adds permissioned requesters and request-level control over who may propose \citep{uma_moov2_2026,uma_moov2_programmatic_2026}.  The current documentation explains the mechanism and present default-whitelist behavior.  It does not by itself establish historical enforcement for each request.  Paper~5 therefore treats documentation, source-code capability, deployed configuration, and historical use as four distinct evidence layers.

Alternative oracle architectures clarify what Paper~5 does not attempt to replace.  Voting-based systems such as Astraea analyze incentive-compatible decentralized truth determination, while foundational oracle-pattern taxonomies distinguish push/pull and inbound/outbound data flows \citep{adler_astraea_2018,muehlberger_oracle_patterns_2020}.  A newer prediction-market literature evaluates automated and hybrid AI--human resolution, including confidence-based escalation for difficult questions \citep{kota_ai_oracles_2026}.  These designs are relevant comparison points, but the present paper measures the deployed Polymarket--UMA--CTF process rather than proposing a substitute oracle.

We use \emph{oracle finality} operationally.  A request is oracle-final when its successful outcome can no longer be changed through the ordinary challenge route and can be consumed by the requester.  This is not a claim about metaphysical truth, legal finality, or blockchain-consensus finality.  It is a measurable mechanism state that can precede adapter consumption and holder redemption.

\subsection{Settlement finality and stage separation}

Traditional financial-market infrastructure separates execution, clearing, settlement, and final transfer.  The Principles for Financial Market Infrastructures define final settlement around an irrevocable and unconditional transfer or discharge of an obligation and emphasize that rules should identify the point at which settlement becomes final \citep{cpss_iosco_2012}.  Paper~5 does not assert that Polymarket or CTF is legally equivalent to a regulated financial market infrastructure.  It adopts the analytical discipline of separating agreement about an obligation from the point at which the obligation is irreversibly discharged.

The event-market analog has at least three post-decision layers: oracle finality, protocol settlement through payout reporting, and holder redemption.  It also has two pre-finality clocks: public evidence sufficiency and contractual decidability.  Mixing all five into a scalar \texttt{resolved\_at} field can misattribute source delay to the oracle, holder inattention to the protocol, or request reset time to the economic question itself.

The distinction also separates two kinds of post-protocol delay.  \emph{Protocol delay} is the time until the claim becomes redeemable.  \emph{Holder delay} is the time until a particular holder chooses to redeem.  The latter can reflect batching, inattention, custody, gas, or portfolio management; it is not automatically an oracle-performance measure.

\subsection{Multistate models, competing risks, and interval-censored clocks}

The resolution path contains both transient and absorbing states.  A first-dispute reset is not a terminal competing event; it is a transition into a new request.  The natural statistical object is therefore a multistate process, with Aalen--Johansen transition-probability estimates and cause-specific or subdistribution models used only for genuinely absorbing destinations \citep{aalen_johansen_1978,fine_gray_1999,andersen_borgan_gill_keiding_1993}.  Kaplan--Meier summaries remain useful for simple marginal durations but cannot represent the full request chain \citep{kaplan_meier_1958}.

Off-chain source and eligibility times introduce a second statistical issue.  A source may be known only to the minute, within a publication interval, or through an archived page with uncertain update time.  Treating such clocks as exact creates artificial precision.  Turnbull's nonparametric framework for interval-censored data supplies the appropriate marginal-duration estimator when exact times are unavailable \citep{turnbull_1976}.  Paper~5 therefore separates an exact on-chain state backbone from source-conditioned duration analyses and reports endpoint sensitivity where regression software requires point times.

Transition-anchored price and liquidity responses draw on event-study methods \citep{mackinlay_1997}.  The application requires three adaptations.  First, proposal timing is endogenous to information arrival, so proposal effects are descriptive unless stronger identification is available.  Second, protocol finality stops trading; finality studies use one-sided convergence and cessation measures rather than a symmetric post-event trading window.  Third, technical state transitions such as question initialization provide negative controls because they should not systematically reveal terminal information when they occur long before the event.

\subsection{Leveraged event-linked markets}

Paper~1 establishes the non-portability of continuous-underlying margin and funding assumptions to binary event underlyings and documents the two design tensions most relevant here: terminal-jump bad debt is not solved by a halt, and dynamic margin can pre-empt more recoveries than it prevents \citep{nechepurenko2026_p1}.  Paper~2 shows that multi-leg variants add asynchronous finality and oracle-composition risk \citep{nechepurenko2026_p2}.  Paper~3 shows that leverage magnifies both market-price and outcome-manipulation incentives \citep{nechepurenko2026_p3}.  Paper~4 establishes which address-level and book-level supply-side measurements are feasible on Polymarket's hybrid central limit order book (CLOB) \citep{nechepurenko2026_p4}.

The present paper adds the missing state variable.  Earlier programme work conditions risk on price, time to scheduled resolution, volatility, liquidity, and participant behavior.  Paper~5 conditions risk on the realized adjudication state, the occupation time in that state, and the remaining path from evidence to protocol-final collateral.  It also shows why an actor-controlled proposal or dispute cannot be used mechanically as a margin or conversion trigger without an explicit manipulation-cost test.

The oracle literature often compares value secured with the resources available to challenge or corrupt a mechanism, but a derivative overlay adds a distinct accounting issue: aggregate payoff can be zero across counterparties while particular participants acquire large state-contingent gains.  Paper~5 therefore separates native collateral and gross exposure from the maximum private terminal-state range.  This connects oracle-security analysis to Paper~3's leverage-amplification argument without treating open interest, net system value, face bond, or private manipulation gain as interchangeable quantities \citep{eskandari_oracles_2021,nechepurenko2026_p3}.

\section{Institutional Architecture of Polymarket Adjudication}
\label{sec:institution}

This section reconstructs the mechanism from frozen documentation, bytecode-matched source,
the deployment registry, and observed historical use.  A current description, a source-code
capability, a deployed configuration, and an observed invocation are different evidence classes.

\subsection{Evidence layers}
\label{sec:institution:evidence}

Four evidence layers must remain separate throughout the paper.

\begin{enumerate}
  \item \textbf{Current documentary description.}  What Polymarket or UMA currently says the mechanism does.
  \item \textbf{Source-code capability.}  What a particular source revision permits, including branches that may never have been invoked.
  \item \textbf{Deployed configuration.}  Which bytecode, dependencies, access controls, bonds, liveness values, and whitelist settings applied to a particular address or request.
  \item \textbf{Observed historical use.}  Which events and calls actually occurred on-chain.
\end{enumerate}

No claim is promoted from one layer to another without evidence.  For example, a manual-resolution function establishes capability, not use; a current default whitelist establishes present membership, not historical membership; and a user-facing statement that anyone may propose does not establish that every historical ManagedOptimisticOracleV2 request had enforcement disabled.

\subsection{Product boundary: prediction markets and Perps}
\label{sec:institution:boundary}

Polymarket currently documents two economically distinct products.  Prediction markets trade outcome tokens that become redeemable according to a discrete event resolution.  Polymarket Perps are non-expiring contracts on external continuous-price underlyings; funding anchors the mark to an index, and there is no event resolution or 0/1 settlement \citep{polymarket_perps_faq_2026}.  The currently documented instruments span indices, commodities, crypto assets, and an equity-linked underlying rather than event probabilities \citep{polymarket_perps_markets_2026}.

The distinction determines the paper's empirical scope.  We do not reverse-engineer a deployed event-probability perpetual because the public product does not contain one.  We measure the existing event-resolution stack and use the current Perps architecture only as a continuous-reference comparison.

\subsection{Question initialization and request creation}
\label{sec:institution:initialization}

At adapter initialization, the question parameters and ancillary data are stored, a two-outcome condition is prepared on CTF, and a price request is sent to the configured Optimistic Oracle \citep{polymarket_uma_adapter_repo_2026,polymarket_uma_adapter_contract_2026}.  The request is marked event-based.  This yields an important distinction:

\begin{equation}
  \trequest \not\equiv \teligible.
\end{equation}

The request is a technical object created by the adapter.  Contractual decidability is a semantic property of the market rules and named source.  A request may exist long before a terminal answer is permissible.  Any attempt to infer decidability from request creation alone will misclassify early proposals and ignore/too-early responses.

\subsection{Rules, ancillary data, and terminal values}
\label{sec:institution:rules}

Polymarket's current resolution documentation states that each market has a named source, end date, and edge-case rules, and that the rules rather than the title govern resolution \citep{polymarket_resolution_docs_2026}.  At the oracle layer, the YES/NO query convention supports three terminal payout values---YES, NO, and UNKNOWN/50--50---plus an ignore or too-early sentinel that is not terminal \citep{uma_polymarket_verification_2026}.

The terminal support is therefore
\[
  Y\in\{0,\tfrac12,1\},
\]
while the ignore value belongs to the transition system rather than the payout support.  Treating ignore as a fourth payoff would overstate terminal ambiguity; treating unknown/50--50 as a failed resolution would understate the actual payout state space.

The proposed value must be retained in the event table.  The first proposal of any value, the first terminal-valued proposal, and the proposal on the successful request can be three different observations.  A proposal that eventually points in the correct direction but was submitted before contractual decidability remains a premature proposal for quality classification.

\subsection{Proposal, liveness, dispute, and request reset}
\label{sec:institution:proposal}

The public flow is a bonded proposal followed by a challenge period.  An undisputed proposal can be accepted after liveness; a challenged proposal begins a further path that can reach the DVM \citep{polymarket_resolution_docs_2026,uma_oracle_overview_2026}.  The adapter implementation adds a specific request-chain structure.  The first dispute callback resets the question and creates a new oracle request with a new request timestamp; a subsequent dispute can engage the DVM backstop \citep{polymarket_uma_adapter_repo_2026,polymarket_uma_adapter_contract_2026}.

A single economic question can therefore generate a sequence
\[
  \mathcal{R}_q=(r_{q1},r_{q2},\ldots,r_{qK}),
\]
where each request has its own creation, proposal, liveness, dispute, and settlement records.  The reset is a transient transition, not a terminal outcome.  The empirical database must retain predecessor/successor links and cannot select only the latest request.

The source-level description that a dispute ``triggers a new proposal round'' and the adapter-level implementation that creates a new request are compatible descriptions at different abstraction levels.  Paper~5 stores both: question-level round and request-level identity.

\subsection{Managed proposal access}
\label{sec:institution:whitelist}

ManagedOptimisticOracleV2 permits a request to carry a proposer whitelist and an enforcement flag.  If enforcement is false, any address may propose; if true, only allowed addresses may propose.  New requests can inherit a default whitelist unless the request manager overrides the configuration \citep{uma_moov2_2026,uma_moov2_programmatic_2026}.  Current UMA documentation also describes eligibility criteria for the managed default list \citep{uma_moov2_whitelist_2026}.

Polymarket's current user-facing page describes both proposal and dispute as open to anyone \citep{polymarket_resolution_docs_2026}.  We do not resolve the difference by choosing one description.  The contract freeze records the surface, while the historical analysis reconstructs request-level enforcement where the evidence permits.  The data model retains:

\begin{itemize}
  \item oracle family and deployed address;
  \item proposal whitelist address and enforcement status at the relevant request;
  \item request-specific versus inherited default configuration;
  \item allowed-proposer set when historically recoverable;
  \item dispute-access restrictions, if any, as a separate field rather than an assumption;
  \item secured market value and proposal/dispute concentration by access regime.
\end{itemize}

On-chain proposer and disputer addresses identify execution entities, not necessarily beneficial owners.  Concentration measures are therefore address-level mechanism measures unless additional entity-link evidence exists.

\subsection{Clarifications as versioned governance events}
\label{sec:institution:clarifications}

Polymarket documentation states that additional context can be published on-chain and should not change the question's fundamental intent \citep{polymarket_resolution_docs_2026}.  The adapter's bulletin-board mixin permits any address to post an update, while its source comment instructs users to consider updates from the question creator \citep{polymarket_bulletin_board_2026}.  The relevant empirical object is therefore not simply ``a clarification exists.''  It is a versioned sequence containing author, timestamp, content hash, and creator-authoritative status.

For each proposal and dispute, the analysis reconstructs the rule corpus as it existed immediately before the transaction.  A clarification published after a proposal cannot be used as if the proposer had observed it; a current UI rendering cannot replace a historical rule state.

\subsection{Oracle availability, adapter consumption, and the Part~II boundary}
\label{sec:institution:consumption}

When a valid oracle value is available, a caller can invoke the adapter's resolution path.  The
standard source retrieves or settles the successful request, constructs the payout vector, calls
CTF, and emits the adapter terminal event in the same transaction
\citep{polymarket_uma_adapter_contract_2026}.  The logical layers remain distinct:

\begin{align*}
  \toracle &: \text{the successful request value is irreversible and consumable},\\
  \tconsume &: \text{the adapter consumes that value},\\
  \tprotocol &: \text{CTF records the payout and makes the condition redeemable}.
\end{align*}

Part I estimates Oracle adjudication and the terminal adapter transition.  The companion
Part II reconstructs the CTF and redemption layers from a separately frozen fixed-snapshot
dataset.  If adapter consumption and payout recording share one transaction, the wall-clock
gap can be zero while the exact log order still identifies distinct state transitions.  Oracle
availability itself can remain interval-observed when no dedicated event marks the first
consumable instant.

\subsection{Pause, flag, reset, and manual resolution}
\label{sec:institution:admin}

The adapter source exposes administrative functions to pause and unpause, flag and unflag, reset, and manually resolve after a safety period under specified conditions \citep{polymarket_uma_adapter_contract_2026,polymarket_adapter_auth_2026}.  The admin set is mutable in source.  These facts establish a deployed governance surface, not historical frequency or motive.

The access-control audit reconstructs current or historical admin state where events permit, while the event ledger measures actual branch use.  Ordinary optimistic-oracle finality, administrative intervention, and manual payout are reported as separate routes.

\subsection{Documented adapter versions and the need for a freeze}
\label{sec:institution:versions}

A user-facing adapter list is not a complete production registry.  At the frozen source boundary used in this study, Polymarket's official operational resolution subgraph enumerates six initialization-emitting adapter data sources: the legacy adapter, standard V2, standard V3.1, standard V4, a NegRisk adapter, and a NegRisk V4 adapter \citep{polymarket_resolution_subgraph_2026}.  The operational source also binds each data source to an application binary interface (ABI), start block, and handlers for initialization, reset, and resolution.

This creates a registry-closure requirement prior to any transport claim.  A chain scan can be complete for the addresses supplied to it and still be incomplete for the venue if the address set omitted an active route.  Paper~5 therefore distinguishes:
\[
  \text{address-scope completeness}
  \quad\text{from}\quad
  \text{venue-registry completeness}.
\]
The production registry is the reconciled union of user-facing documentation, official operational indexer configuration, deployment/source artifacts, and on-chain code/dependency state.  Documentary labels remain evidence about naming; operational data sources are evidence about expected historical coverage; bytecode and logs determine deployed use.  G-REGISTRY-CLOSURE must pass before initialized-question-weighted mapping coverage or path frequencies are reported.

\subsection{Architecture summary}
\label{sec:institution:summary}

The deployed system is best represented as two linked processes:

\begin{enumerate}
  \item a \emph{question process} carrying rules, clarifications, administrative state, terminal payout, and CTF condition; and
  \item a \emph{request process} carrying oracle creation, proposal value, liveness, dispute, reset, and settlement.
\end{enumerate}

The question process can outlive an individual request, and request resets can occur without changing the economic question.  This two-level architecture is the formal object of \Cref{sec:formal} and the database design of \Cref{sec:data}.

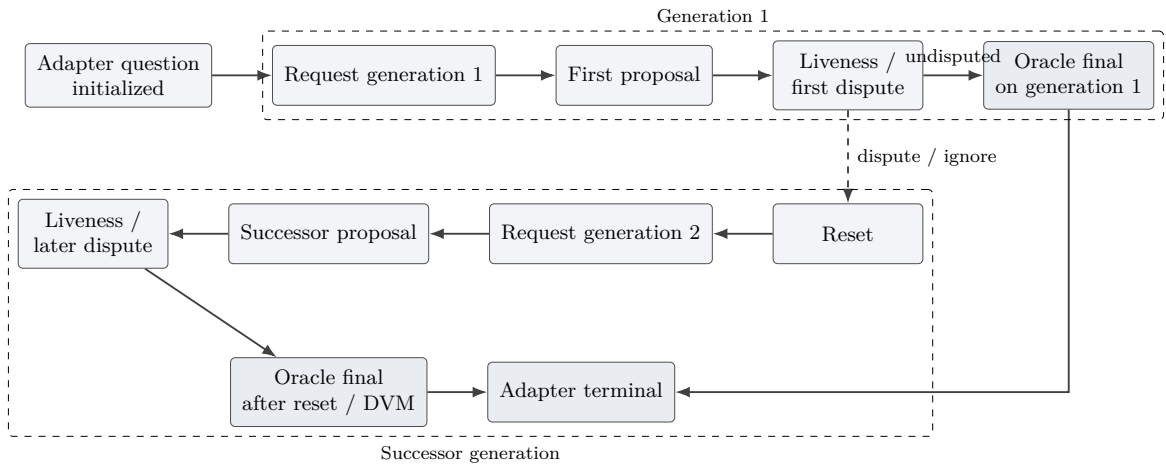
\begin{figure}[htbp]
\centering
\resizebox{0.98\textwidth}{!}{%
\begin{tikzpicture}[
  node distance=7mm and 9mm,
  every node/.style={font=\footnotesize},
  st/.style={seriesbox,minimum width=2.25cm,minimum height=9mm},
  terminal/.style={seriesbox,minimum width=2.25cm,minimum height=9mm,fill=seriesblue!10},
  arr/.style={seriesarrow},
  alt/.style={seriesarrow,dashed}
]
\node[st] (init) {Adapter question\\initialized};
\node[st,right=of init] (req1) {Request generation 1};
\node[st,right=of req1] (prop1) {First proposal};
\node[st,right=of prop1] (live1) {Liveness /\\first dispute};
\node[terminal,right=of live1] (settle1) {Oracle final\\on generation 1};
\node[st,below=14mm of live1] (reset) {Reset};
\node[st,left=of reset] (req2) {Request generation 2};
\node[st,left=of req2] (prop2) {Successor proposal};
\node[st,left=of prop2] (live2) {Liveness /\\later dispute};
\node[terminal,below=14mm of prop2] (settle2) {Oracle final\\after reset / DVM};
\node[terminal,right=of settle2] (adapter) {Adapter terminal};

\draw[arr] (init) -- (req1);
\draw[arr] (req1) -- (prop1);
\draw[arr] (prop1) -- (live1);
\draw[arr] (live1) -- node[above,font=\scriptsize] {undisputed} (settle1);
\draw[alt] (live1) -- node[right,font=\scriptsize] {dispute / ignore} (reset);
\draw[arr] (reset) -- (req2);
\draw[arr] (req2) -- (prop2);
\draw[arr] (prop2) -- (live2);
\draw[arr] (live2) -- (settle2);
\draw[arr] (settle1) |- (adapter);
\draw[arr] (settle2) -- (adapter);

\node[draw,dashed,rounded corners=2pt,fit=(req1)(prop1)(live1)(settle1),inner sep=4pt,label={[font=\scriptsize]above:Generation 1}] {};
\node[draw,dashed,rounded corners=2pt,fit=(reset)(req2)(prop2)(live2)(settle2),inner sep=4pt,label={[font=\scriptsize]below:Successor generation}] {};
\end{tikzpicture}%
}
\caption{Two-level request-generation state machine.  A dispute or ignore/too-early branch can create a successor request without changing the economic question.  Therefore question-level state alone is insufficient for path reconstruction; request generation is a load-bearing state variable.}
\label{fig:request_state_machine}
\end{figure}

\section{Formal Framework: Questions, Requests, Rules, and Clocks}
\label{sec:formal}

\subsection{Economic questions, rule versions, and request chains}

Let \(q\in\mathcal Q\) denote an adapter-level economic question and let \(\rho_q(t)\) denote
the rule state visible at time \(t\), including creator-authoritative clarifications already
recorded by that time.  Let
\begin{equation}
  \mathcal R_q=(r_{q1},\ldots,r_{qK_q})
  \label{eq:request_chain}
\end{equation}
be the ordered chain of oracle request generations associated with \(q\).  A reset can increase
\(K_q\) without changing the economic question.

The terminal YES payout is
\begin{equation}
  Y_q\in\left\{0,\tfrac12,1\right\}.
  \label{eq:payout_support}
\end{equation}
The value \(1/2\) is a terminal split/unknown payout.  Ignore and too-early values belong to
the transition system rather than the terminal support.

For request \(r_{qk}\), let \(v_{qk}\) denote the proposed value.  Define
\begin{align}
  \tproposalany(q)
  &=\min_k t^{(k)}_{\mathrm{proposal}},\\
  \tproposalterminal(q)
  &=\min_{k:\,v_{qk}\in\{0,1/2,1\}}
    t^{(k)}_{\mathrm{proposal}},
\end{align}
when the relevant set is non-empty.  The first proposal of any value, the first
terminal-valued proposal, and the proposal on the successful request can therefore be three
different events.

\subsection{Clock observation as points, ordered events, or intervals}

The mechanism contains exact chain events, block-bounded technical clocks, and potentially
interval-observed semantic clocks:
\begin{align*}
  \tevent &: \text{rule-relevant real-world occurrence or completion},\\
  \tsource &: \text{earliest reproducible publication of sufficient evidence},\\
  \teligible &: \text{earliest time the contemporaneous rule permits a terminal answer},\\
  \trequest &: \text{technical request creation},\\
  \tproposal &: \text{proposal submission},\\
  \tdispute &: \text{dispute submission},\\
  \treset &: \text{reset and successor-request creation},\\
  \toracle &: \text{irreversible consumable oracle outcome},\\
  \tadapter &: \text{terminal adapter transition},\\
  \tclarify &: \text{creator-authoritative on-chain clarification}.
\end{align*}

When a semantic clock is not exact, it is represented by a closed interval:
\begin{align}
  \sourceinterval&=[\underline t^S_q,\overline t^S_q],\\
  \eligibleinterval&=[\underline t^E_q,\overline t^E_q].
  \label{eq:semantic_intervals}
\end{align}
The decision-ready interval is
\begin{equation}
  \readyinterval=
  \left[
    \max\{\underline t^S_q,\underline t^E_q\},
    \max\{\overline t^S_q,\overline t^E_q\}
  \right].
  \label{eq:ready_interval}
\end{equation}
A proposal before the lower endpoint is unambiguously early; a proposal inside the interval is
timing-ambiguous; and a proposal after the upper endpoint is unambiguously post-readiness.
Midpoint substitution is not a primary estimator.

Transactions in one block are ordered by transaction and log index even when their wall-clock
timestamp is equal.  Part I records that exact secondary order.  Missing semantic clocks are
not reconstructed from request, proposal, settlement, adapter, metadata-retrieval, or
clarification timestamps.

\begin{definition}[Technical request context]
\label{def:technical-request-context}
A technical request context begins at creation of request generation \(r_{qk}\).  It establishes
that a request exists; it does not by itself establish evidence sufficiency, contractual
decidability, or active adjudication.
\end{definition}

\begin{definition}[Active adjudication generation]
\label{def:active-adjudication-generation}
An active adjudication generation is a request generation with a proposal or post-dispute
successor path for which liveness, challenge, reset, later adjudication, or settlement transitions
are being evaluated under the contemporaneous rule state.
\end{definition}

\begin{definition}[Oracle finality]
\label{def:oracle-finality}
Oracle finality occurs when the successful request outcome can no longer be changed through
the ordinary challenge path and is consumable by the requester under the deployed mechanism.
\end{definition}

\begin{definition}[Adapter terminal transition]
\label{def:adapter-terminal-transition}
The adapter terminal transition is the first verified terminal transition emitted by the
requesting adapter for question \(q\).  It is the terminal boundary of Part I, not a synonym for
protocol payout recording or holder redemption.
\end{definition}

\subsection{Routes, payout subtype, and non-exclusive path flags}

A single route label cannot encode every relevant path feature.  We use:
\begin{enumerate}
  \item \(\route\), the terminal route;
  \item \(Y_q\), the terminal payout subtype; and
  \item \(\pathflags\), non-exclusive flags for dispute, reset, clarification, pause/flag,
        access restriction, order indeterminacy, and non-standard behavior.
\end{enumerate}

The route set is
\begin{equation}
  \route\in\{G_1,G_2,G_D,G_M,G_C\},
\end{equation}
where \(G_1\) is first-generation oracle finality, \(G_2\) is successor-generation finality,
\(G_D\) is later-adjudication/DVM finality, \(G_M\) is administrative/manual finality, and
\(G_C\) is right-censored or unresolved non-standard status.  This representation prevents a
reset from being treated as an absorbing outcome and prevents a split payout from being
conflated with failure.

\subsection{Two-level state process}

Let \(A_q(t)\) denote the question-level state and \(R_{qk}(t)\) the state of request
generation \(k\).  A request generation can follow:
\begin{equation}
\text{Requested}\rightarrow\text{Proposed}\rightarrow
\begin{cases}
\text{OracleFinal},\\
\text{Disputed}\rightarrow\text{Reset/NewRequest},\\
\text{Disputed}\rightarrow\text{LaterAdjudication}\rightarrow\text{OracleFinal},\\
\text{Ignore/TooEarly}\rightarrow\text{Reset/NewRequest}.
\end{cases}
\label{eq:request_path}
\end{equation}
The question-level process additionally carries rule-version marks, clarifications,
administrative states, and adapter terminality.

The process can be described by transition intensities
\begin{equation}
\lambda_{ij}(t\mid X_q,\rho_q(t))
=\lim_{\Delta t\downarrow0}
\frac{\Prob\{S_q(t+\Delta t)=j\mid S_q(t)=i,X_q,\rho_q(t)\}}{\Delta t}.
\label{eq:transition_intensity}
\end{equation}
This notation does not assert a structural homogeneous Markov model; it states that duration
and destination must be modeled jointly.

\begin{proposition}[Request history is required for a Markov state]
\label{prop:request_history}
Suppose a dispute on the first request can create a successor request while a dispute on the
successor can expose a different later-adjudication branch.  A coarse state representation that
retains only the current question state but omits active request generation and equivalent reset
history is not Markov: two admissible histories can have the same coarse state and different
feasible next transitions.
\end{proposition}

For transient state \(j\), define occupation over \([a,b]\):
\begin{equation}
  \occupancy_{qj}(a,b)=\int_a^b\ind{S_q(t)=j}\,dt.
  \label{eq:state_occupation}
\end{equation}
Occupation time can affect funding, collateral, liquidity exposure, and operational cost even
when two paths share the same terminal payout and terminal adapter timestamp.

The following result extends the path-projection logic of Paper~1.  That paper shows that
terminal account risk cannot be inferred from a continuous-underlying margin proxy; the
mechanism-specific result below shows that a terminal payout and timestamp cannot recover
adjudication-state occupation.

\begin{proposition}[A scalar terminal timestamp is insufficient]
\label{prop:scalar_timestamp}
Let two admissible paths have the same terminal payout and the same adapter terminal timestamp
but different state-occupation vectors.  For any path-dependent functional
\begin{equation}
  Z_q(P)=z_0(Y_q,\tadapter)
  +\sum_{j\in\mathcal J}h_j\occupancy_{qj}(P),
  \label{eq:path_functional}
\end{equation}
if the occupation-vector difference is not orthogonal to
\(\mathbf h=\bigl(h_j\bigr)_{j\in\mathcal J}\), no function of the scalar pair
\((Y_q,\tadapter)\) can recover \(Z_q(P)\) for both paths.
\end{proposition}

\subsection{Latency and occupation estimands}

For exact technical clocks, relevant mechanism intervals include:
\begin{align}
 L_{qk}^{RP}&=t_{qk}^{\mathrm{proposal}}-t_{qk}^{\mathrm{request}},\\
 L_{qk}^{PD}&=t_{qk}^{\mathrm{dispute}}-t_{qk}^{\mathrm{proposal}},\\
 L_{qk}^{PS}&=t_{qk}^{\mathrm{settle}}-t_{qk}^{\mathrm{proposal}},\\
 L_q^{OA,Q}&=\tadapter-\toracle.
\end{align}
Reset-to-successor-proposal is measured on the successor generation rather than merged into
the original request age.

When source and eligibility are observed, semantic durations are interval-censored relative to
\(\readyinterval\).  When they are unavailable, the technical intervals remain valid mechanism
measurements but do not become truth-to-resolution latency.

The on-chain adjudication duration decomposes through state occupation:
\begin{equation}
  \tadapter-\trequestcreated
  =
  \sum_{j\in\mathcal J_{\mathrm{onchain}}}
  \occupancy_{qj}(\trequestcreated,\tadapter).
  \label{eq:onchain_occupation_decomposition}
\end{equation}

\subsection{Proposal quality and convergence}

Proposal quality is coded at request level:
\begin{equation}
Q_{qk}\in
\{\mathrm{premature},\mathrm{terminal\ correct},
\mathrm{terminal\ incorrect},\mathrm{unknown},\mathrm{censored}\}.
\label{eq:proposal_quality}
\end{equation}
A proposal is premature when contemporaneous rules require ignore or continuation, even if
its direction later equals the terminal payout.  A disputed proposal is not automatically
incorrect, and ex-post directional agreement does not establish timeliness.

Let \(p_{q,t}\in[0,1]\) be an observed YES-token price and \(Y_q\) the terminal payout.  The
terminal convergence error is
\begin{equation}
  e_{q,t}=|p_{q,t}-Y_q|.
  \label{eq:terminal_error}
\end{equation}
This is a hindsight-aligned descriptive quantity.  It can reflect rule ambiguity, disagreement,
capital lock-up, fees, inventory constraints, stale quotes, and finality delay; it is not
interpreted as a pure oracle premium without a stronger identification design.

\subsection{Observation grades}
Every clock or transition receives one coordinate-specific grade:
\begin{description}[style=nextline,leftmargin=2.6cm]
\item[$\ObsExact$ --- Exact] Exact event or primary-source time with deterministic identity and, where relevant, exact chain order.
\item[$\ObsInterval$ --- Interval] The target coordinate is bounded by exact lower and upper observations but is not point-identified.
\item[$\ObsSnapshot$ --- Snapshot] A current or terminal representation is observed without the historical transition time.
\item[$\ObsProxy$ --- Proxy] A documented substitute is available but is not functionally identical to the target coordinate.
\item[$\ObsUnmeasured$ --- Unmeasured] Available evidence supports neither a point nor an interval.
\item[$\ObsConflict$ --- Conflicting] Load-bearing sources disagree; the coordinate is quarantined from primary estimation.
\end{description}
Grades are clock-specific rather than question-wide.

\section{Data, Population Construction, and Identification}
\label{sec:data}

\subsection{Relational unit of observation}

The database does not use ``market'' as a universal key.  It distinguishes the economic
event, venue listing, adapter question, CTF condition, oracle request generation, and each
proposal/dispute/settlement event.  The frozen denominator is an adapter-question population,
not automatically a unique economic-market population across migration, wrapping, or duplicate
listing surfaces.

\subsection{Event sourcing and bitemporal provenance}

The analysis is event-sourced.  Current state is derived from an append-only sequence rather
than treated as the only historical record.  Every accepted transition retains
\begin{equation}
(\text{chain ID},\text{block hash},\text{transaction hash},\text{log index})
\label{eq:canonical_event_identity}
\end{equation}
together with source contract, decoded payload, ABI/source hash, and processing lineage.
Valid time and retrieval/processing time remain separate.  A clarification, metadata object,
or configuration observed later is not projected backward without historical evidence.

\subsection{Frozen adapter population}

The preserved adapter archive contains 351,734 raw logs.  After the cutoff at block
79,721,080, 350,703 logs remain.  Canonical identity, duplicate/conflict checks, ABI decoding,
and a network-disabled rebuild establish the accepted adapter layer.

The denominator contains 185,550 initialized adapter questions.  Questions without an observed
terminal adapter event at the boundary are right-censored.

\subsection{Registry closure and version-aware linkage}

Transport completeness is conditional on the registry supplied to the scanner.  The registry is
closed before population inference through deployment/source artifacts, bytecode/proxy
relationships, official operational indexer configuration, ABI/event topics, and observed
adapter-address/topic pairs.

Adapter-to-oracle linkage is version-aware.  Each request generation is identified through
exact requester, identifier, timestamp, ancillary-data, and family-specific fields.  No
nearest-time, pooled-question, or question-ID-only fallback is used.

\subsection{Oracle transport and exact lifecycle reconstruction}

Oracle events are acquired through exact filters:
\begin{equation}
  \text{oracle address}
  +\text{event topic0}
  +\text{indexed requester adapter}.
  \label{eq:oracle_filter}
\end{equation}
Ranges are recursively split until every accepted terminal response is below the provider cap.
The accepted lifecycle layer covers 38 streams and 808 terminal ranges.  Thirty-five exact
archive-provider comparisons reproduce selected identities and payloads.

The lifecycle table contains \code{RequestPrice}, \code{ProposePrice},
\code{DisputePrice}, and \code{Settle}.  It preserves exact ABI integers and immutable event
identity.  Sixteen order-indeterminate generations remain in the canonical data but are excluded
from estimands requiring a unique generation order.

\subsection{Rules, clarifications, and metadata}

Every cutoff question has an initial on-chain rule record.  Updates are stored as versioned marks
with author, timestamp, content hash, and creator-authoritative status.  Proposal and dispute
records join to the latest valid rule state strictly before their exact chain order.

Metadata is joined only by exact stable identifier.  Re-parsing nested Gamma event pages
raises exact coverage to 104,032 of 185,550 adapter-questions (56.07\%); 81,518 remain
unmatched and none are ambiguous under the exact-ID rule.  Unmatched questions remain
unmatched rather than being title-, slug-, date-, or time-matched.  Because coverage is strongly
route-dependent, metadata-conditioned results report route composition and do not treat the
matched cohort as missing at random.

\subsection{Source and evidence hierarchy}

Evidence is ordered as historical on-chain state; bytecode-matched pinned source; signed
deployment/governance record; frozen official documentation; mutable current documentation;
and secondary reporting.  Mechanism statements are tagged as documented, capable, configured,
observed, or inferred.

\subsection{Validity gates}

\begin{description}[style=nextline,leftmargin=2.2cm]
  \item[G-REGISTRY] Adapter, oracle, dependency, ABI, and start-block registry closes.
  \item[G-ADAPTER] Adapter payload integrity, cutoff coverage, decoding, and denominator pass.
  \item[G-ORACLE] Request-generation linkage, lifecycle coverage, and canonical identity pass.
  \item[G-RULE] Contemporaneous rule versions and authoritative updates are recoverable for the
        stated cohort.
  \item[G-METADATA] Exact stable-ID metadata coverage is reported with route-dependent missingness.
  \item[G-CLOCK] Every clock used in a duration claim is exact, block-bounded, interval-observed,
        or explicitly unmeasured.
  \item[G-ACCESS] Historical access comparisons are limited to requests with observed configuration.
  \item[G-REBUILD] Reported analytical artifacts reproduce under a network-disabled rebuild.
\end{description}

A failed gate narrows the claim; it does not authorize a substitute observable chosen after
results are known.  The final network-disabled rebuild compares 45 reported analytical artifacts
and reports zero failures.

\subsection{Missingness and censoring}

No terminal event by the cutoff is right censoring, not permanent non-resolution.  Legacy
questions for which the modern request event is not applicable are
\code{NOT\_APPLICABLE\_NO\_REQUEST}.  Missing oracle evidence is not an undisputed path.
Unknown/split payout is a terminal subtype, not a data failure.  Reset is transient.

\subsection{Boundary to the downstream protocol dataset}

Part I's accepted population and estimands stop at the adapter terminal transition.  The
companion Part II uses a separately frozen Conditional Tokens event population to study payout
recording and redemption.  Completion of that downstream acquisition does not alter Part I's
question, request-generation, rule, or clarification denominators, and no downstream aggregate
is inserted here as a substitute for an Oracle-layer result.

\subsection{Cost and storage discipline}

The historical extraction exposed a methodological failure: sparse address-only queries against
a time-partitioned public warehouse can scan tens of terabytes while returning a small research
table.  The accepted dataset retains the resulting evidence and audit record, but the acquisition
lesson is general.  Registry closure, ABI-derived event selection, dry-run byte limits, and
claim-targeted extraction precede expensive transport.

The reproducibility companion records query coordinates, provider, range ledger, input/output bytes,
row counts, schemas, hashes, and manual interventions.  Acquisition history belongs in the
reproducibility record rather than the main empirical narrative.

\subsection{Ethics and disclosure}

The analysis uses public on-chain identifiers but does not attempt off-chain deanonymization.
Address concentration is reported as execution-address concentration.  Aggregate reporting avoids ranked address lists where aggregation is sufficient.  The author discloses that the analysis forms part of a broader event-linked leverage research
programme; this paper does not claim production safety for an associated design.

\section{Estimands and Empirical Design}
\label{sec:estimands}

\subsection{Nested observation universes}

The paper uses the full frozen adapter-question population, the subset with exact oracle
request-generation linkage, and estimand-specific subsets satisfying the required rule, clock,
access, or metadata gate.  Every result reports its denominator and attrition.

\subsection{Population and path census}

Experiment E1 constructs adapter-family counts, initialized-question denominator, terminal and
right-censored adapter states, exact request-generation multiplicity, proposal/dispute/reset
paths, terminal payout subtype, and non-exclusive path flags.  These are population facts.

\subsection{On-chain multistate estimands}

The exact backbone supplies request-to-first-proposal, proposal-to-first-dispute,
proposal-to-ordinary threshold, reset-to-successor proposal, proposal/dispute-to-settlement,
and oracle-to-adapter intervals where identifiable.  Right-censored generations remain in the
risk set.

\subsection{Semantic-clock estimands}

When source and eligibility intervals are observed, proposal and finality are measured relative
to \(\readyinterval\).  In the accepted population-wide layer, those semantic clocks remain
unmeasured.  The corresponding result is negative identification: technical mechanism clocks
cannot substitute for missing external clocks.

\subsection{Rule and clarification estimands}

The semantic layer measures initial rule coverage, update authorship and timing, clarification
order relative to request/proposal/dispute/reset/settlement, and exact metadata attrition.
Clarification timing establishes order, not causal or retroactive effect.

\subsection{Access, concentration, and proposal quality}

Proposal and dispute addresses are measured separately.  Concentration statistics include top
shares and Herfindahl--Hirschman index.  Historical access comparisons are limited to requests
with observed configuration.  Proposal quality distinguishes premature/ignore, terminal
correct, terminal incorrect, split/unknown, and censored categories.

\subsection{Transition-anchored microstructure boundary}

The integrated master specified price, spread, depth, staleness, and convergence studies around
proposal and finality transitions.  Part I retains the estimand definitions and negative controls
but does not promote unexecuted market-window analyses into results.  A nonexistent
post-terminal order book is not coded as zero spread or depth.

\subsection{Version-locked hypotheses}

The original integrated paper fixed directional hypotheses before the final analysis layers.
After the split, their disposition is explicit in \Cref{app:hypotheses}.  Hypotheses requiring
external clocks, price windows, value denominators, or mechanism replay are marked unmeasured,
transferred, or pending rather than rewritten after inspection.

\subsection{Uncertainty and multiplicity}

Population counts are exact.  Duration distributions use medians, quantiles, restricted means,
or interval-aware estimators as appropriate.  Secondary associations cluster by economic event
where multiple questions share one real-world event.  Multiple-testing adjustment applies to
families of secondary transition outcomes, not to primary population facts.

\section{Empirical Results I: Adapter Population and State Paths}
\label{sec:results_adapter}

This section reports the adapter population and adapter-emitted state paths through the frozen historical cutoff.  Proposal, dispute, later-adjudication, access-control, and oracle-settlement events are joined in \cref{sec:results_oracle}.  The unit is an adapter-address/question-ID pair.  It is not automatically a unique economic market across adapter migration or wrapping.

\subsection{Source preservation, canonical identity, and cutoff}
\label{sec:results_adapter:coverage}

The frozen source archive contains 80 Parquet exports with 351,734 raw adapter logs.  Every payload passed size, checksum, schema, and row-count checks.  One file retains a cloud-generation provenance warning despite payload equivalence; the local SHA-256 manifest is therefore the post-cloud identity root.  An independent content-addressed copy-and-reopen audit reproduced all 80 files and all 351,734 rows.

Canonicalization uses $(\text{chain id},\text{block hash},\text{transaction hash},\text{log index})$.  The canonicalization audit found 351,734 active unique chain logs, no removed rows, no duplicate collapse, and no conflicting canonical identity.  The cutoff analysis includes 350,703 logs at or below Polygon block 79,721,080.  The maximum observed cutoff log is block 79,721,066 at 2025-11-30 23:59:30 UTC.  NegRisk V4 is excluded because its operational activation is later than the cutoff.  The other five adapter families have completed monthly transport evidence through the boundary.

The decoder observed 32 adapter-address/topic pairs.  All 350,703 cutoff events decoded successfully under the frozen ABI registry; there were no unknown pairs and no decode failures.  A clean network-disabled rebuild compared 34 outputs and found no byte-hash mismatch.  These results establish the adapter-log population backbone but do not establish oracle-event completeness.

\subsection{Initialized adapter-question denominator}
\label{sec:results_adapter:population}

The cutoff denominator contains 185,550 initialized adapter questions.  \Cref{tab:adapter_population} reports the population by adapter family.  Terminal-observed and unresolved columns are states at the cutoff.  They do not measure ultimate resolution success or failure because questions initialized near the boundary have less time to reach a terminal adapter event.

\begin{table}[htbp]
\centering
\caption{Initialized adapter-question population through Polygon block 79,721,080.  ``Terminal'' is an adapter-log terminal event by the cutoff; ``right-censored'' is not an ultimate failure classification.  Counts reproduce the manifest-locked adapter population output.}
\label{tab:adapter_population}
\scriptsize
\begin{tabular}{@{}llrrrr@{}}
\toprule
Adapter & Route & Initialized & Resets & Terminal & Right-censored \\
\midrule
NegRiskUmaCtfAdapter & NegRisk & 76,271 & 606 & 59,234 & 17,037 \\
UmaCtfAdapterOld & Standard & 5,829 & 73 & 2,941 & 2,888 \\
UmaCtfAdapterV2 & Standard & 28,600 & 590 & 26,900 & 1,700 \\
UmaCtfAdapterV31 & Standard & 13,379 & 85 & 12,667 & 712 \\
UmaCtfAdapterV4 & Standard & 61,471 & 124 & 54,378 & 7,093 \\
\bottomrule
\end{tabular}
\end{table}

The terminal-state attrition table contains 156,116 \texttt{QuestionResolved} paths, four \texttt{QuestionSettled} paths without a later adapter terminal event in the cutoff ledger, and 29,430 questions unresolved at the cutoff.  The last category is treated as right-censored in the exact request-generation analysis.  In particular, it must not be divided by the denominator and described as a ``non-resolution rate'' without an age-at-cutoff risk set.

\subsection{Observed event families and adapter paths}
\label{sec:results_adapter:paths}

\Cref{tab:event_inventory} reports the decoded event-family counts.  The population contains 1,478 reset events and 1,570 ancillary-data-update events, alongside 17 pause-or-flag events and ten administrative events.  These are observed invocations, unlike source-code capabilities that may never have been exercised.

\begin{table}[htbp]
\centering
\caption{Decoded adapter-event inventory at the historical cutoff.  All 32 observed adapter-address/topic pairs decoded under frozen ABIs; no unknown topic or decode failure remained.}
\label{tab:event_inventory}
\small
\begin{tabular}{lr}
\toprule
Event classification & Event count \\
\midrule
Initialization & 185,550 \\
Resolution or settlement & 159,048 \\
Dependency or non-state & 3,030 \\
Ancillary-data update & 1,570 \\
Reset & 1,478 \\
Pause or flag & 17 \\
Administrative & 10 \\
\bottomrule
\end{tabular}
\end{table}

The adapter-level state surface is dominated by a direct initialization-to-resolution path in both standard and NegRisk routes, but the ledger also retains request-forwarding, reset, ancillary-update, pause/flag, and unresolved paths.  \Cref{tab:top_adapter_paths} reports the most frequent exact sequences.  A reset is a transient event, not an absorbing outcome; the same question key can therefore have a longer successor-request path once oracle events are joined.

\begin{table}[htbp]
\centering
\caption{Most frequent adapter-state paths through the cutoff.  The table reports exact path counts from the manifest-locked adapter output; oracle events are reported separately and are not included in these adapter-only path counts.}
\label{tab:top_adapter_paths}
\small
\begin{tabular}{llr}
\toprule
Route & Adapter-event sequence & Questions \\
\midrule
Standard & Initialized $\rightarrow$ Resolved & 92,964 \\
NegRisk & Initialized $\rightarrow$ Resolved & 57,919 \\
NegRisk & Initialized; no later adapter event by cutoff & 16,856 \\
Standard & Initialized; no later adapter event by cutoff & 12,263 \\
Standard & Initialized $\rightarrow$ ResolutionDataRequested $\rightarrow$ Settled $\rightarrow$ Resolved & 2,853 \\
Standard & Initialized $\rightarrow$ Reset $\rightarrow$ Resolved & 626 \\
NegRisk & Initialized $\rightarrow$ AncillaryDataUpdated $\rightarrow$ Resolved & 615 \\
NegRisk & Initialized $\rightarrow$ Reset $\rightarrow$ Resolved & 513 \\
\bottomrule
\end{tabular}
\end{table}

The adapter-only sequence cannot tell whether a reset followed an incorrect proposal, a too-early value, an administrative intervention, or another oracle branch.  Nor does it reveal who proposed or disputed, when the ordinary liveness threshold became irreversible, or when a DVM result became available.  Those quantities require the oracle lifecycle rather than inference from adapter order alone.

\subsection{Censoring and transition to the Oracle layer}
\label{sec:results_adapter:censoring}

The historical cutoff is a research boundary, not the end of every question's economic life.  The 29,430 questions without an observed terminal adapter event therefore enter the Oracle analysis as right-censored question-level paths rather than as ultimate failures.  The adapter layer fixes the denominator and exact reset/terminal evidence; \cref{sec:results_oracle} then joins raw \texttt{RequestPrice}, \texttt{ProposePrice}, \texttt{DisputePrice}, and \texttt{Settle} events to request generations.

This separation is substantive.  The adapter population provides complete question-level coverage through a declared boundary, while the Oracle layer supplies request-generation multiplicity, first-proposal timing, dispute/reset branches, ordinary liveness thresholds, DVM-inclusive intervals, and address-level governance.  Neither layer is allowed to impute the clocks or states that belong to the other.

\section{Empirical Results II: Oracle Requests, Disputes, and Finality Paths}
\label{sec:results_oracle}

This section reports the second manifest-locked empirical layer.  It joins the complete adapter-question population of \cref{sec:results_adapter} to raw events emitted by the legacy Optimistic Oracle, OptimisticOracleV2, and ManagedOptimisticOracleV2.  The cutoff is independently re-derived as Polygon block 79,721,080.  The unit below is either an adapter question or an exact Oracle request generation, depending on the estimand.  No nearest-time, pooled-question, or question-ID-only fallback is used.

The accepted Oracle analysis is qualified in two ways.  First, branches that pass through a second dispute are \emph{DVM interval-only}: the data observe the second dispute and subsequent Oracle \texttt{Settle}, but not an exact DVM result time or value.  Second, Managed Optimistic Oracle (Managed OO) access history is \emph{partial observed-only}: two whitelist-update events are retained, but no access state is projected backward to requests for which the historical configuration is not observed.  These qualifications bound the claims; they do not invalidate the raw lifecycle census.

\subsection{Transport, decoding, and exact request linkage}
\label{sec:results_oracle:coverage}

The acquisition uses exact filters of the form Oracle address plus lifecycle-event topic plus indexed requester adapter.  Etherscan range leaves are recursively split until every terminal range is below the response cap.  The resulting zero-gap ledger covers 38 streams and 808 terminal ranges.  A bounded archive remote procedure call (RPC) validation compares 35 exact ranges---thirty positive and five terminal-zero---and finds no identity or payload discrepancy.

\begin{table}[htbp]
\centering
\caption{Oracle lifecycle acquisition and decoding through the verified cutoff.}
\label{tab:oracle_inventory}
\begin{tabularx}{\textwidth}{@{}Y r Y@{}}
\toprule
Layer / event & Count & Interpretation \\
\midrule
Canonical raw Oracle rows & 504,334 & Unique canonical identities; no identity conflict or removed duplicate. \\
Decoded lifecycle events & 504,332 & Full ABI decode for the four request-lifecycle event families. \\
\texttt{RequestPrice} & 184,148 & Exact request generations linked to adapter questions. \\
\texttt{ProposePrice} & 159,447 & All observed proposal events; generation-order analyses exclude sixteen order-indeterminate cases. \\
\texttt{DisputePrice} & 1,604 & Observed dispute events across the legacy Optimistic Oracle, Optimistic Oracle V2, and Managed Optimistic Oracle routes. \\
\texttt{Settle} & 159,133 & Oracle settlement/consumption events; not automatically adapter or Conditional Tokens protocol finality. \\
Managed Optimistic Oracle access/configuration events & 2 & Two observed default-proposer-whitelist updates; historical access state is not back-projected. \\
Exact Etherscan streams & 38 & Oracle address + event topic + requester topic for lifecycle streams, plus access/configuration streams. \\
Terminal range leaves & 808 & Zero-gap recursive block-range ledger. \\
Independent archive-provider overlap ranges & 35 & Thirty positive and five terminal-zero exact comparisons, all matching. \\
\bottomrule
\end{tabularx}
\end{table}

The 504,334 canonical raw rows contain 504,332 lifecycle events and two Managed OO access/configuration events.  Canonical identity is
\[
(\text{chain id},\text{block hash},\text{transaction hash},\text{log index}).
\]
No canonical identity conflict and no removed duplicate is observed.  Exact decoding yields 184,148 request generations, 159,447 proposal events, 1,604 dispute events, and 159,133 Oracle settlement events.

The exact linkage layer relates 184,148 request generations to the frozen 185,550 adapter-question population. The accounting identity closes exactly: 182,671 questions have at least one Oracle request and 1,477 successor generations produce the remaining request rows. It preserves 2,879 legacy questions as \texttt{NOT\_APPLICABLE\_NO\_REQUEST}; they are structural zero-request observations in the frozen state layer, neither failed joins nor Oracle-risk-set censoring.  Sixteen order-indeterminate generations are retained in the canonical data but excluded from generation-order estimands.  The historical legacy Finder mapping remains unresolved and is never replaced with current-state inference.

\subsection{Technical request creation is not resolution commencement}
\label{sec:results_oracle:request_clock}

A central empirical result is that the technical request timestamp is not a general measure of when a market became economically or contractually resolvable.  Modern adapters create or forward an event-based request at initialization.  The first proposal can arrive much later, after evidence and rule conditions develop.  By contrast, after a first-dispute reset, the successor request already belongs to an active adjudication path and is typically followed by a new proposal much more quickly.

\begin{table}[htbp]
\centering
\caption{Kaplan--Meier median durations in seconds. The legacy initialization-to-adapter-terminal median is computed on all 5,829 initialized legacy questions, with 2,941 terminal events and 2,888 right-censored observations.  The request clock is technical request creation, not source publication or contractual decidability.  Decentralized-verification-mechanism values are settlement upper bounds rather than exact result times.}
\label{tab:oracle_duration_medians}
\scriptsize
\resizebox{\textwidth}{!}{%
\begin{tabular}{lrrrrr}
\toprule
Adapter family &
\makecell[r]{Request to\\first proposal} &
\makecell[r]{Ordinary threshold\\to Oracle settle} &
\makecell[r]{Initialization to\\adapter terminal} &
\makecell[r]{Reset to successor\\first proposal} &
\makecell[r]{Second dispute to\\settle upper bound} \\
\midrule
Legacy adapter & 182 & 196 & 80,033,650 & -- & -- \\
Standard V2 & 264,165 & 147 & 282,655 & 2,713 & 258,868 \\
Standard V3.1 & 206,927 & 84 & 213,193 & 934 & 322,682 \\
Standard V4 & 176,388 & 598 & 180,836 & 300 & 261,266 \\
NegRisk adapter & 744,151 & 246 & 790,510 & 2,909 & 242,262 \\
\bottomrule
\end{tabular}}
\end{table}

The median request-to-first-proposal duration is 182 seconds in the legacy route, but 176,388 seconds for Standard V4, 206,927 seconds for Standard V3.1, 264,165 seconds for Standard V2, and 744,151 seconds for the original NegRisk adapter.  Those durations mix technical request age with unobserved evidence and contractual-decidability time; they must not be called ``resolution delay.''  The corresponding reset-to-successor-proposal medians are 300, 934, 2,713, and 2,909 seconds on the modern reset-capable routes.  The contrast is consistent with the two-level state representation, but it is not a causal test of that architecture: the initial request can be created long before an event becomes answerable, whereas reset occurs inside an already active adjudication path. The existence of exact successor generations and reset-linked request identities, rather than the magnitude contrast alone, is the direct empirical support for the two-level representation.

The semantic layer therefore measures source, rule, and clarification clocks rather than using request age as a semantic proxy. Until those clocks pass the registered rule and source gates, request-to-proposal duration is reported as a mechanism interval only.

\subsection{First proposal, dispute, reset, and ordinary settlement}
\label{sec:results_oracle:paths}

First-proposal paths are modeled with competing risks.  A proposal can receive a first dispute, reach its ordinary expiration threshold, or remain right-censored at the cutoff.  Missing transitions are never folded into the ordinary branch.

\begin{table}[htbp]
\centering
\caption{First-proposal competing paths by adapter family.  ``No first dispute by threshold'' denotes that the proposal liveness threshold elapsed without an observed first dispute. It is not a terminal-state classification and does not identify a separate ignore/too-early mechanism.  Legacy uses the original Optimistic Oracle; V2, V3.1, and NegRisk use Optimistic Oracle V2; V4 uses Managed Optimistic Oracle V2.}
\label{tab:first_proposal_competing}
\footnotesize
\begin{tabular}{lrrrr}
\toprule
Adapter family & Risk set & First dispute & No first dispute by threshold & Right-censored \\
\midrule
Legacy adapter & 3,006 & 22 & 2,984 & 0 \\
Standard V2 & 27,831 & 711 & 27,120 & 0 \\
Standard V3.1 & 12,916 & 102 & 12,810 & 4 \\
Standard V4 & 55,132 & 144 & 54,781 & 207 \\
NegRisk adapter & 60,546 & 625 & 59,900 & 21 \\
\bottomrule
\end{tabular}
\end{table}

The observed first-proposal risk sets contain 1,604 first disputes.  The mechanism after a first dispute differs by adapter generation.  Modern V2, V3.1, V4, and NegRisk paths usually produce an exact reset linked to a successor request; the legacy route instead proceeds through its own settlement semantics.

\begin{table}[htbp]
\centering
\caption{Observed path after the first dispute.  \texttt{SETTLE\_UPPER\_BOUND} is an observed Oracle \texttt{Settle} endpoint when no exact decentralized-verification-mechanism result timestamp is available.  The reset/successor column is conditional on an observed first dispute. It does not claim that all adapter reset events are dispute-driven; legacy and other non-dispute reset causes remain separate accounting objects. Modern dispute-triggered reset links are exact, while further escalation remains interval-qualified.}
\label{tab:first_dispute_paths}
\footnotesize
\begin{tabular}{lrrrr}
\toprule
Adapter family & First disputes & Reset / successor request & \texttt{SETTLE\_UPPER\_BOUND} & Right-censored \\
\midrule
Legacy adapter & 22 & 0 & 22 & 0 \\
Standard V2 & 711 & 588 & 123 & 0 \\
Standard V3.1 & 102 & 85 & 17 & 0 \\
Standard V4 & 144 & 124 & 18 & 2 \\
NegRisk adapter & 625 & 606 & 18 & 1 \\
\bottomrule
\end{tabular}
\end{table}

Among modern reset-capable paths, the reset is transient rather than terminal.  The successor request is separately identified, and its proposal clock begins at the reset/request generation rather than at the original question initialization.  A later dispute can enter a DVM-inclusive branch.  For those branches the observed endpoint is the subsequent Oracle \texttt{Settle}, coded as an upper bound on the unavailable exact DVM result time.

The ordinary undisputed branch exhibits a different scale.  Conditional on first-proposal expiration without dispute, median time from the ordinary finality threshold to observed Oracle \texttt{Settle} ranges from 84 seconds on Standard V3.1 to 598 seconds on Standard V4.  That short interval is an Oracle-consumption delay, not the full event-to-cash delay.  It begins only after the proposal's liveness has elapsed and ends before holder redemption.

For second-dispute branches, median time from second dispute to Oracle \texttt{Settle} upper bound ranges from 242,262 to 322,682 seconds across the observed modern routes. The term ``second dispute is request-generation specific: it denotes a first dispute observed on a successor generation (generation number at least one), not a second dispute inside one Oracle request. The frozen successor-generation risk set contains 253 such disputed generations.  The result supports a materially longer disputed path, but the exact DVM voting/result clock remains unidentified and is not midpoint-imputed.

\subsection{Oracle settlement and adapter consumption}
\label{sec:results_oracle:consumption}

Oracle \texttt{Settle} is not automatically protocol settlement finality.  A local transaction-adjacency audit finds 75,284 Oracle settle rows with a strictly later adapter terminal log in the same transaction.  This establishes an atomic consumption path for that observed subset.  Non-adjacent settles are not treated as missing or failed: some adapter paths consume an already-settled request in a later transaction, and the current Oracle-only layer does not identify a request-specific ordinary-finality-threshold-to-adapter-consumption endpoint for every generation.

Accordingly, Paper~5 reports three distinct objects:

\begin{enumerate}
  \item the ordinary proposal threshold becoming irreversible under the observed request path;
  \item the Oracle \texttt{Settle} event;
  \item the adapter/CTF terminal transition, where observed.
\end{enumerate}

The companion protocol-finality study links these states to Conditional Tokens payout
reporting and observed redemption using a separately frozen dataset.  It does not fill a missing
cross-domain duration with same-block, nearest-time, or transaction-level guesses.

\subsection{Address-level governance, proposal quality, and access limits}
\label{sec:results_oracle:governance}

Proposal and dispute actors are observed as execution addresses.  The data do not identify beneficial owners or establish that two addresses belong to distinct economic entities.  Concentration measures are therefore address-level execution concentration.

\begin{table}[htbp]
\centering
\caption{Observed Managed Optimistic Oracle proposer concentration and settled-only first-proposal agreement for Standard V4.  Concentration is address-level, not beneficial-owner concentration.  Agreement is exact ABI-decoded integer equality on the eligible settled subset.}
\label{tab:managed_governance}
\scriptsize
\resizebox{\textwidth}{!}{%
\begin{tabular}{lrrrrrrr}
\toprule
Month & Proposals & Unique proposers & Top-1 share & Top-5 share & HHI & Eligible settled & Exact agreement \\
\midrule
2025-08 & 11 & 6 & 0.272727 & 0.909091 & 0.190083 & 11 & 0.909091 \\
2025-09 & 13,707 & 87 & 0.664770 & 0.854673 & 0.454111 & 13,707 & 0.997592 \\
2025-10 & 14,036 & 67 & 0.373112 & 0.693859 & 0.172417 & 14,036 & 0.998290 \\
2025-11 & 27,378 & 89 & 0.453138 & 0.707320 & 0.228846 & 27,102 & 0.998967 \\
\bottomrule
\end{tabular}}
\end{table}

The Standard V4 Managed OO cohort illustrates a governance asymmetry.  In September 2025 the largest observed proposer address accounts for 0.664770 of proposal events and the top five for 0.854673; in November the corresponding shares are 0.453138 and 0.707320.  Exact first-proposal agreement with the eventual settled price is 0.997592, 0.998290, and 0.998967 in the September--November settled-only cohorts.  These are not unconditional accuracy estimates: disputed, unsettled, order-indeterminate, and missing-price cases follow the recorded exclusions.

Observed disputers are more diffuse in the high-volume Managed OO cohort.  In November 2025, 64 dispute events are distributed across 39 addresses, with top-one share 0.125000, top-five share 0.312500, and HHI 0.040527.  The comparison is descriptive and period-specific; it does not establish an access-policy causal effect.

Historical access configuration remains the binding governance limitation.  The separate access/configuration scan observes two \texttt{DefaultProposerWhitelistUpdated} events and no requester-access, \texttt{RoleGranted}, or \texttt{RoleRevoked} event in the scanned event surface.  The 61,595 Managed OO request rows are therefore classified as requester/proposer access unresolved rather than retrospectively allowed or denied.  No current whitelist or role getter is projected backward.

\subsection{Rebuild and claim boundary}
\label{sec:results_oracle:rebuild}

A controlled materialization and a clean network-disabled rebuild compare three population files, six analysis files, and seven governance files.  All reported analytical files match by SHA--256; Parquet schemas and row counts also match.  A content-addressed backup and extraction audit independently verifies the frozen Oracle source archive.

\begin{table}[htbp]
\centering
\caption{Post-review reconciliation of load-bearing Part I counts. The final column states the strongest claim supported by the frozen evidence.}
\label{tab:review_reconciliation}
\footnotesize
\begin{tabularx}{\textwidth}{@{}L{3.3cm}rY@{}}
\toprule
Object & Count / estimate & Disposition \\
\midrule
Initialized adapter-question instances & 185,550 & Frozen population denominator. \\
Questions with at least one Oracle request & 182,671 & Exact question-level count. \\
Request generations & 184,148 & Equals 182,671 first generations plus 1,477 successor generations. \\
Legacy initialization-to-adapter-terminal KM & 80,033,650 s & Median reached on 5,829 initialized legacy questions; 2,941 events and 2,888 censored. \\
Questions with two or more resets & 75 & Exact cohort: 71 have two resets, three have three, and one has four. The common causal explanation of all historical residual-75 identities is not established. \\
Residual-75 paths with observed dispute evidence & 72 & At least one observed dispute appears in the retained request-generation history; three residual-cohort questions have reset evidence without an observed dispute. \\
Rule-declared end-date proxy & NOT EVALUATED & The frozen corpus contains no separately labelled rule-end-date field that can be promoted to contractual decidability. \\
Adapter--Oracle relation & Identity join & Exact question/request linkage is established; causal event-to-event ordering is not inferred from the key join alone. \\
\bottomrule
\end{tabularx}
\end{table}

The exact 75-question multi-reset cohort is retained as a diagnostic accounting object rather than assigned a mechanism by coincidence. Seventy-two of the 75 questions have observed dispute evidence in at least one request generation; three have reset evidence without an observed dispute. No literal \texttt{Ignore/TooEarly} sentinel is present in the frozen event inventory, so that mechanism is not inferred. The same discipline applies to rule end dates: a \texttt{resolutionTime} parameter is not relabelled as contractual decidability.

The empirical claim boundary is now precise.  Paper~5 can report exact request generations, observed proposal/dispute/settle paths, right-censoring-aware mechanism durations, address-level concentration, settled-only proposal agreement, and the observed Managed OO configuration timeline.  It cannot report an exact DVM result time or value, historical access state before observed updates, beneficial-owner concentration, a resolved historical legacy Finder timeline, or a request-specific ordinary-finality-to-adapter-consumption duration where the endpoint is not identified.

These results establish the Oracle-finality layer of the paper.  Rules and external-source
clocks remain separately graded in Part I; Conditional Tokens protocol finality and holder
redemption are reported in Part II rather than inferred from the Oracle data.  Transition-anchored
market microstructure remains a separate empirical design.

\section{Rules, Metadata, and the On-Chain Clarification Clock}
\label{sec:results_semantic}

The adapter and Oracle layers identify how a question moves through the mechanism.  They do not, by themselves, identify which rule text traders faced, whether a venue record can be joined to a human-readable market, or when an authoritative clarification entered the on-chain record.  The accepted semantic-layer extraction adds those relations while preserving the distinction between a measured mechanism clock and an unmeasured external-evidence clock.

\subsection{Population-complete initial rules and observed updates}
\label{sec:results_semantic:rules}

Every one of the 185,550 cutoff adapter questions has an initial ancillary-data rule record.  The versioned corpus contains 185,550 initial versions and 1,570 observed on-chain update versions, for 187,120 rule versions in total.  Exact UTF-8 decoding succeeds for 187,119 versions; one decode error remains in the raw corpus.  Of the update records, 1,549 are creator-authoritative under the frozen identity rule, 12 are exact but non-authoritative, and nine remain unclassifiable because the frozen identity is missing or non-exact.  The creator-authoritative records affect 1,350 distinct adapter questions.

The source parser retains 194,124 literal URL mentions over 818 domains.  These counts are not a complete named-source taxonomy: one rule version can contain multiple URLs, many rules name sources in prose rather than by URL, and the historical webpage state behind a URL is generally not frozen.  The on-chain corpus is therefore population-complete as a record of initial ancillary data and observed updates, but only partially informative about historical off-chain source content.

\begin{table}[htbp]
\centering
\caption{Rule, metadata, and clock coverage through the verified cutoff. Rates use the denominator stated in the final column; update records are event counts, not unique questions.}
\label{tab:semantic_coverage}
\small
\begin{tabularx}{\textwidth}{YrrY}
\toprule
Layer & Count & Rate & Denominator / qualification \\
\midrule
Initialized adapter questions & 185{,}550 & 100.00\% & cutoff population \\
Gamma exact stable-identifier match, baseline & 654 & 0.35\% & cutoff population \\
Gamma exact stable-identifier match, nested union & 104{,}032 & 56.07\% & cutoff population \\
Gamma exact-identifier unmatched & 81{,}518 & 43.93\% & cutoff population \\
Gamma exact-identifier ambiguous & 0 & 0.00\% & cutoff population \\
Initial on-chain rule versions & 185{,}550 & 100.00\% & cutoff population \\
Observed on-chain update versions & 1{,}570 & --- & update records \\
Exact update timestamps & 1{,}561 & 99.43\% & update records \\
Creator-authoritative exact updates & 1{,}549 & 98.66\% & update records \\
Non-authoritative exact updates & 12 & 0.76\% & update records \\
Unclassifiable update identities & 9 & 0.57\% & update records \\
External-source clock: \texttt{U\_UNMEASURED} & 185{,}550 & 100.00\% & cutoff population \\
Contractual-decidability clock: \texttt{U\_UNMEASURED} & 185{,}550 & 100.00\% & cutoff population \\
\bottomrule
\end{tabularx}
\end{table}

\subsection{Nested Gamma records materially improve exact metadata coverage}
\label{sec:results_semantic:gamma}

The first Gamma pass used the standalone markets transport and produced only 654 exact stable-ID matches.  Re-parsing the already frozen event pages reveals 1,815,777 nested market records.  Their append-only union with 59,173 standalone market rows contains 1,874,950 transport rows and raises the exact question match to 104,032 of 185,550 questions (56.07\%), with 81,518 unmatched and no ambiguous exact matches.  The gain is entirely from frozen data; no fuzzy title, slug, date, or nearest-time rule is introduced.

Coverage is strongly route-dependent.  Exact-match rates are 98.94\% for Standard V4, 97.78\% for V3.1, 95.54\% for V2, and 48.12\% for the legacy adapter, while the original NegRisk adapter has no exact stable-ID Gamma match under the frozen identifier surface.  The unmatched residual is therefore not plausibly missing at random.  Metadata-conditioned analyses must report adapter composition and cannot silently generalize the 56\% exact cohort to the full population.

The metadata gate remains partial despite the large gain.  Terminal cursors and exact-ID accounting pass, but repeat-boundary response pages were not frozen.  The result is a closed exact relation to the available frozen transports, not proof that Gamma exposed every historical market record at every retrieval boundary.

\subsection{The clarification clock is measured on-chain}
\label{sec:results_semantic:clarification}

Exact chain identity---chain, adapter address, block number, transaction hash, and log index---recovers canonical timestamps for 1,561 of the 1,570 accepted update records.  This includes all 1,549 creator-authoritative records and 12 non-authoritative records; the nine unclassifiable identities remain explicit, and no conflicting identity is observed.  The timestamp is an on-chain clarification/update clock.  It is not relabelled as external publication time or as the instant at which the contract became decidable.

Joining creator-authoritative updates to the exact request-generation paths produces 1,711 clarification--request-generation rows.  The multiplicity is expected: one question can have more than one request generation after a reset.  \Cref{tab:clarification_ordering} shows that 1,426 rows (83.34\%) fall after the technical request but before the first proposal.  A further 99 rows occur after an exact reset and before the successor proposal.  These observations provide direct population-level evidence that the rule-governance surface can change while a request already exists and that request creation is not a sufficient semantic start time for adjudication.

\begin{table}[htbp]
\centering
\caption{Position of creator-authoritative on-chain clarifications relative to exact request-generation paths.  The unit is a clarification--request-generation relation; one clarification can map to more than one generation.}
\label{tab:clarification_ordering}
\begin{tabular}{lrr}
\toprule
Ordering class & Rows & Share \\
\midrule
Before first request & 58 & 3.39\% \\
Request to first proposal & 1{,}426 & 83.34\% \\
After first proposal, before first dispute/threshold & 17 & 0.99\% \\
After reset, before successor proposal & 99 & 5.79\% \\
After successor proposal & 60 & 3.51\% \\
After Oracle \texttt{Settle} & 51 & 2.98\% \\
\midrule
Total exact request-generation rows & 1{,}711 & 100.00\% \\
\bottomrule
\end{tabular}
\end{table}

The remaining ordering classes require care.  Sixty rows place a clarification after a successor proposal and 51 after an Oracle \texttt{Settle} for the linked request generation.  These rows establish order, not legal or economic effect: the data do not show that an already settled request was retroactively changed, and multiple generations can coexist within one question history.  They are retained as governance-path observations rather than folded into a single ``pre-resolution clarification'' category.

Signed duration tables are available only where both named endpoints are exact.  Initialization-to-clarification duration remains unavailable because no exact initialization endpoint is available; request time is not substituted for it.  Negative signed values are retained where the declared endpoint order makes them meaningful; no midpoint imputation is used.

\subsection{A population-level negative result for external clocks}
\label{sec:results_semantic:external_clocks}

The frozen evidence does not identify defensible population-level publication times for named external sources.  All 185,550 external-source clocks and all 185,550 contractual-decidability clocks remain \texttt{U\_UNMEASURED}.  Current Gamma retrieval fields, technical request creation, proposal, dispute, Oracle settlement, adapter terminal, and clarification block time are not used as proxies.

This negative result changes the interpretation of the Oracle timing results.  Request-to-proposal duration is an interval between two mechanism events; it is not ``evidence delay'' or ``resolution delay.''  Proposal timing relative to external publication, and whether a proposal was early or late under contemporaneous facts, remain unidentifiable at population scale.  The paper can measure rule versions and on-chain clarifications exactly while still being unable to observe the first moment at which the external world supplied a terminal answer.

\subsection{Implications for a leveraged overlay}
\label{sec:results_semantic:leverage}

The semantic layer reinforces the state-contingent design argument in three ways.  First, the technical-request state is too coarse for margin or funding triggers: most measured clarifications occur after request creation but before proposal.  Second, a creator-authoritative clarification is a distinct, verifiable governance transition that can change the interpretation risk of an open position even when no price or Oracle event occurs at the same moment.  Third, the absence of a population-wide external clock prevents a venue from calibrating leverage compression to a researcher's reconstructed ``truth time''; any live design must use observable protocol states, explicitly sourced venue attestations, or conservative policy windows rather than a hidden semantic timestamp.

These results do not show that clarifications cause disputes, price moves, or delay.  Clarified questions are selected, and the present analysis does not join market-microstructure windows.  The causal and price-discovery questions therefore remain outside this paper; the present contribution is a measured rule-governance clock and an explicit boundary on what it is not.

\section{Discussion and Implications for Leveraged Design}
\label{sec:design_implications}

\subsection{The empirical object is a state surface, not a countdown}

The results do not show that every event market needs more margin whenever a request is old.
They show that request age combines several economically different situations.  A question can
be technically requested but not yet semantically decidable; it can have a live undisputed
proposal; it can be inside a first-dispute reset; or it can be awaiting later adjudication.  A
single countdown from initialization cannot distinguish those states.

A later risk engine needs at least the adapter-question identity, active request generation,
proposal value and liveness state, dispute/reset history, contemporaneous rule version,
observation grade, oracle settlement, adapter terminality, and the downstream Part~II state.

\begin{designprinciple}[State variables must follow the mechanism]
\label{dp:state-variables-must-follow-the-mechanism}
A risk control should not use one scalar ``resolution time'' when feasible next actions and
remaining exposure depend on request generation, dispute history, rule state, and finality layer.
\end{designprinciple}

\subsection{State-triggered collateral and leverage}

The reconstruction supports an architecture in which collateral and leverage change at verified
state transitions rather than only through realized volatility or scheduled time.  Candidate
triggers include first terminal-valued proposal, first dispute, reset, successor proposal, ordinary
liveness expiry, oracle finality, and adapter terminality.

Part I does not select the optimal trigger.  Paper~6 tests whether state-triggered controls can
reduce the pre-emption trade-off identified in Paper~1 without worsening bad debt.

\subsection{Actor control and trigger manipulability}

A transition is not safe as a risk trigger merely because it is observable on-chain.  If an actor
with a large event-linked position can cause or accelerate a proposal, dispute, or reset, the
control can transfer value.  A trigger must therefore be evaluated against actor access,
external state-contingent exposure, bond/fee/penalty, and value transferred through margin,
liquidation, auction, or conversion.

Paper~3 supplies the manipulation taxonomy and Paper~6 formalizes the accounting test.
Part I establishes the historical actor/access surface but does not infer intent or beneficial
ownership.

\subsection{Funding during long-lived request contexts}

Modern request-to-proposal medians are measured in days, but the request clock can start before
semantic readiness.  Funding that scales mechanically with request age can therefore transfer
value before the event is contractually decidable.

A defensible design distinguishes basis funding supported by an executable convergence loop,
finality carry for capital locked in a verified post-result state, custody/service charges, and
policy transfers unrelated to basis convergence.  Hidden research-only clocks are unsuitable
live inputs.

\subsection{Rule updates and governance}

A technical request can predate a creator-authoritative rule update that later shapes
adjudication.  A system that freezes rule state at initialization would misrepresent the observed
governance surface.  A system that uses the latest text without versioning would back-project
later clarifications into earlier proposals.

The minimum auditable design stores the initial rule, each authoritative update, exact ordering
relative to proposals/disputes, source/content hash, and an explicit policy for prospective or
interpretive effect.  The empirical ordering alone does not decide retroactivity.

\subsection{Oracle finality is not protocol or cash finality}

Part I stops at the adapter boundary.  A successful oracle value can be irreversible yet not
have been consumed; an adapter can be terminal while payout recording or holder redemption
remains a separate observation.  Part II reports those downstream states from its separately frozen protocol dataset.

This boundary matters for margin release, conversion, and cash-equivalent valuation.  A design
cannot release all collateral merely because an oracle request is settled if delivery, protocol
payout, or custody remains unresolved.  Holder inaction after redeemability should not be
mislabeled as oracle delay.

\subsection{Implications for Papers~6 and~9}

Paper~6 receives state/trigger definitions, exact request paths, durations, actor-control
variables, rule-version marks, and finality endpoints.  Paper~9 receives the same objects at leg
level and studies asynchronous finality.  Detailed conversion accounting, auctions, collateral
envelopes, and shared-capacity constraints are intentionally not duplicated here.

\section{Limitations}
\label{sec:limitations}

\paragraph{Population unit.}
The population is complete for the frozen adapter-address/question-ID scope, not necessarily
for unique economic markets across migrations, wrappers, or duplicated listings.

\paragraph{Cutoff censoring.}
Questions and generations without a terminal event at the cutoff are right-censored.  They do
not establish permanent non-resolution.

\paragraph{External semantic clocks.}
External-source publication and contractual decidability remain unmeasured population-wide.
The paper cannot estimate truth-to-proposal or truth-to-finality latency for the full cohort.

\paragraph{Later-dispute timing.}
Exact DVM result time and value are not observed.  Dispute-to-later settlement is an inclusive
interval or upper bound, not a point estimate of voting duration.

\paragraph{Access and identity.}
Historical Managed Optimistic Oracle access state is partial observed-only.  Proposal/dispute
concentration is address-level; beneficial ownership, relayers, and common control are not
identified.

\paragraph{Clarification interpretation.}
Clarification timing establishes order and authorship, not causal, legal, or retroactive effect.
The on-chain rule corpus does not reconstruct every historical webpage or named-source update.

\paragraph{Metadata selection.}
Exact Gamma coverage is incomplete and route-dependent.  Metadata-conditioned results apply
to the exact matched cohort.

\paragraph{Proposal quality.}
Exact final-value agreement is not an unconditional accuracy score and does not establish
timeliness under contemporaneous rules.

\paragraph{Administrative rationale.}
On-chain evidence can show pause, flag, reset, manual handling, proposal, dispute, and
settlement.  It does not reveal motive or private communication.

\paragraph{Price and liquidity transitions.}
Part I does not promote unexecuted price/book analyses as results.  Public fills do not recover
the address-level quote lifecycle established as unavailable in Paper~4.

\paragraph{Downstream finality.}
Part I stops at adapter terminality by design.  CTF payout recording and observed redemption
are estimated in the companion Part II and are not pooled into Part I.  Neither analysis claims total entitlement redeemed, unredeemed value, or complete holder cash realization
without the corresponding entitlement denominator.

\paragraph{Venue scope.}
Results are Polymarket-specific.  Kalshi comparison requires native state reconstruction and
aligned endpoints rather than label matching.

\paragraph{Design scope.}
State-triggered margin, provisional settlement, auction, and claim conversion are mechanism
families evaluated in Paper~6, not deployable specifications proved here.

\paragraph{Data disclosure boundary.}
The article reports aggregate event, timing, and governance results; raw address-level and restricted reconstruction materials are not distributed.

\paragraph{Reset-cause identification.}
The frozen state and event layers identify an exact cohort of 75 questions with two or more resets, but they do not establish one common causal mechanism for every arithmetic residual of the same cardinality. In particular, no distinct \texttt{Ignore/TooEarly} event literal is observed, so the paper does not attribute the cohort to that branch.

\paragraph{Rule-declared end dates.}
A bounded audit found no separately labelled rule-end-date field that can be interpreted as contractual decidability across the frozen population. The observed \texttt{resolutionTime} parameter is not promoted to that semantic clock. The end-date proxy analysis is therefore reported as not evaluated rather than filled by a technical timestamp.

\section{Conclusion}
\label{sec:conclusion}

Polymarket adjudication is a versioned, multistate process rather than a scalar timestamp.
Through Polygon block 79,721,080, the accepted adapter layer contains 350,703 decoded cutoff
logs and 185,550 initialized adapter questions.  Exact Oracle extraction adds 504,332 lifecycle events comprising 184,148 requests, 159,447 proposals, 1,604 disputes, and 159,133 settlements.
Request generations are linked without nearest-time or question-ID-only fallback.

The central timing result separates technical context from active adjudication.  Modern
request-to-first-proposal medians are measured in days, while reset-to-successor-proposal
medians are measured in minutes or tens of minutes.  The first interval can begin before
sufficient evidence, contractual decidability, or even a creator-authoritative clarification.
The second begins after a contested proposal inside an active path.

The rule layer reaches the same conclusion from a different direction.  Every cutoff question
has an initial on-chain rule version, the versioned ledger contains 1,570 updates, and 1,549
creator-authoritative updates receive exact timestamps.  Among 1,711 exact
clarification--generation relations, 1,426 occur after request creation but before first
proposal.

At the same time, external-source publication and contractual-decidability time remain
unmeasured population-wide.  Current metadata, proposal time, oracle settlement, adapter
terminality, and clarification block time are not substitutes.  The paper reconstructs
governance while refusing to fabricate a truth-to-resolution clock.

For leveraged design, request generation, dispute path, rule version, state occupation, actor
access, and observation grade are balance-sheet inputs.  A control that sees only scheduled
resolution time or current volatility misses whether a question is merely requested, clarified,
proposed, disputed, reset, oracle-final, or adapter-terminal.  Paper~6 evaluates the resulting
mechanism choices; Paper~9 extends them to asynchronous multi-leg finality.

Post-review reconciliation strengthens the accounting boundary. The legacy initialization-to-adapter-terminal Kaplan--Meier median is 80,033,650 seconds on the full 5,829-question legacy cohort. The request-generation identity closes exactly at 182,671 request-bearing questions plus 1,477 successor generations. An exact 75-question multi-reset cohort is retained, but the frozen event inventory does not justify assigning one common \texttt{Ignore/TooEarly} cause; the paper reports that non-identification rather than converting a cardinality match into a causal claim.

The companion Part II follows the downstream protocol path through Conditional Tokens payout
recording and observed redemption using a separately frozen fixed-snapshot dataset.  The split
preserves the central result: adjudication, protocol settlement, and holder realization are
different states, and each is measurable only to the extent that its own evidence and denominator
support it.

\appendix
\section{Operational Entity, State, and Clock Dictionary}
\label{app:dictionary}

The data model separates economic questions, oracle request generations, protocol conditions, and holder redemptions.  The distinction prevents a venue-facing ``market'' identifier or a final timestamp from silently replacing technically different objects.

\subsection{Entity hierarchy}

\begin{table}[htbp]
\centering
\caption{Canonical entity hierarchy used across Parts I and II.}
\label{tab:sup_entity_ontology}
\small
\begin{tabularx}{\textwidth}{@{}L{2.7cm}L{4.2cm}Y@{}}
\toprule
Entity & Canonical key & Analytical role \\
\midrule
Economic question & Adapter address plus question identifier & Rule versions, clarifications, request chain, terminal adapter path \\
Oracle request generation & Request-family-specific tuple plus deterministic hash & Proposal, liveness, dispute, reset, later adjudication, oracle settlement \\
Proposal or dispute & Transaction hash plus log index & Actor, value, bond, timing, and challenge path \\
Adapter terminal event & Adapter address, question ID, transaction hash, log index & Part I endpoint and bridge to downstream protocol state \\
Conditional Tokens condition & Condition identifier derived from oracle, question ID, and slot count & Payout numerators, denominator, and protocol redeemability in the companion Part II \\
Holder redemption & Condition, redeemer, transaction hash, log index & Observed collateral realization in the companion Part II \\
\bottomrule
\end{tabularx}
\end{table}

The successful request is not assumed to be the latest request by timestamp.  It is the request whose value is actually consumed by the adapter or whose verified exceptional branch produces the terminal adapter outcome.

\subsection{Question-level states}

\begin{table}[htbp]
\centering
\caption{Question-level states and attached clocks.}
\label{tab:sup_question_states}
\small
\begin{tabularx}{\textwidth}{@{}L{3.0cm}L{2.5cm}Y@{}}
\toprule
State or clock & Observation type & Operational definition \\
\midrule
Initialized & Exact on-chain & Adapter initialization for an exact adapter-question key \\
Evidence sufficient & Point, interval, or unmeasured & Named source contains sufficient rule-relevant evidence \\
Contractually decidable & Point, interval, or unmeasured & Contemporaneous rule permits a terminal answer \\
Clarification mark & Exact on-chain & Creator-authoritative or other versioned rule update; not terminal \\
Paused or flagged & Exact on-chain where observed & Administrative safety state; dormant source-code capability is not coded as observed \\
Oracle final & Exact or interval-qualified & Successful request is irreversible under the ordinary challenge path and consumable \\
Adapter terminal & Exact on-chain & Terminal transition emitted by the requesting adapter \\
Protocol final & Exact on-chain, the companion Part II & Conditional Tokens payout vector recorded \\
Holder redeemed & Exact on-chain per holder, the companion Part II & Successful redemption linked to condition and redeemer \\
\bottomrule
\end{tabularx}
\end{table}

\subsection{Request-generation states}

\begin{table}[htbp]
\centering
\caption{Request-generation states.}
\label{tab:sup_request_states}
\small
\begin{tabularx}{\textwidth}{@{}L{3.0cm}L{2.5cm}Y@{}}
\toprule
State & Observation type & Operational definition \\
\midrule
Request created & Exact on-chain & Requester, identifier, request fields, ancillary data, and linked adapter question \\
Awaiting proposal & Derived exact interval & Request creation until proposal, reset, non-standard termination, or censoring \\
Proposed-any & Exact on-chain & Any proposed value, including ignore or too-early \\
Proposed-terminal & Exact on-chain & Proposed value belongs to the supported terminal payout set \\
Liveness & Derived interval & Proposal until dispute, ordinary expiry, replacement, or censoring \\
Disputed & Exact on-chain & Verified dispute event and actor \\
Reset & Exact on-chain & Adapter reset and successor-request relation; transient, not absorbing \\
Later adjudication & Exact/interval & Later-dispute branch whose result time may be interval-observed \\
Oracle settled & Exact/interval & Request value set or first verifiably consumable \\
Superseded & Derived & Request no longer governs after exact successor creation \\
Right-censored & Cutoff exact & No observed next state by the locked boundary \\
\bottomrule
\end{tabularx}
\end{table}

\subsection{Ordering and rule versions}

On-chain records are canonically ordered by block number and log index, with transaction hash as a stable tie-break and transaction index retained as auxiliary metadata where available.  Events in one block can be ordered but do not have separately observed wall-clock duration.  A reverted transaction produces no state transition.

For question $q$, the contemporaneous rule process is a right-continuous marked sequence
\begin{equation}
  \rho_q(t)=\rho_{q0}+\sum_{j=1}^{J_q}\Delta\rho_{qj}\ind{t\ge t^{\mathrm{clarify}}_{qj}}.
\end{equation}
Each mark stores exact author classification, chain identity, raw content or content hash, and observation status.  A proposal is paired with the last verified rule version preceding the proposal in canonical chain order.  Later text is never backfilled.

\subsection{Observation grades}
\begin{description}[style=nextline,leftmargin=2.8cm]
\item[$\ObsExact$ --- Exact] Exact on-chain event, exact primary-source timestamp, or deterministic state transition.
\item[$\ObsInterval$ --- Interval] A bounded interval from exact primary evidence.
\item[$\ObsSnapshot$ --- Snapshot] Current or terminal state observed without transition history.
\item[$\ObsProxy$ --- Proxy] Documented but functionally imperfect substitute.
\item[$\ObsUnmeasured$ --- Unmeasured] Unsupported clock; excluded from primary latency estimation.
\item[$\ObsConflict$ --- Conflicting] Load-bearing evidence disagrees.
\end{description}
Grades are clock-specific rather than question-wide. A question can have an $\ObsExact$ adapter-terminal time and an $\ObsUnmeasured$ external-source time.

\section{Estimands, Denominators, and Gate Consequences}
\label{app:estimands}

\subsection{Population and path estimands}

Primary population estimands include:

\begin{itemize}
  \item adapter questions by route and cutoff state;
  \item request generations per adapter question;
  \item exact proposal, dispute, reset, settlement, and adapter-terminal path counts;
  \item right-censoring counts and risk-set age;
  \item rule versions and clarification events;
  \item exact metadata match and attrition counts.
\end{itemize}

Every percentage names its denominator.  Event counts are not silently converted into unique-question counts.

\subsection{Duration estimands}

Mechanism intervals include:

\begin{align*}
  D^{RP}_{qk} &= t^{(k)}_{\mathrm{proposal}}-t^{(k)}_{\mathrm{request}},\\
  D^{RS}_{qk} &= t^{(k+1)}_{\mathrm{proposal}}-t^{(k)}_{\mathrm{reset}},\\
  D^{OS}_{qk} &= t^{(k)}_{\mathrm{settle}}-t^{(k)}_{\mathrm{ordinary\ threshold}},\\
  D^{AT}_{q} &= \tadapter-t^{\mathrm{init}}_q.
\end{align*}

Kaplan--Meier and competing-risk summaries follow the standard distinction between right censoring and alternative destinations \citep{kaplan_meier_1958,aalen_johansen_1978}.  The first quantity is not evidence delay unless source and decidability clocks are independently observed.  Later-dispute durations are represented by exact intervals or upper bounds when the data-verification-mechanism result timestamp is absent.

\subsection{Governance estimands}

Proposal and dispute concentration is computed over execution addresses, using top shares and the Herfindahl--Hirschman index.  Results are not described as beneficial-owner concentration.  Proposal agreement is measured only on the exact eligible settled subset and reports its denominator and attrition.

Clarification ordering uses the relation between one creator-authoritative update and one exact request generation.  One update can therefore contribute several rows when the question has several generations.

\subsection{Gate consequence matrix}

\begin{table}[htbp]
\centering
\caption{Mandatory claim consequence of a failed gate.}
\label{tab:sup_gate_consequences}
\small
\begin{tabularx}{\textwidth}{@{}L{2.0cm}L{4.0cm}Y@{}}
\toprule
Gate & Failure & Consequence \\
\midrule
Population & Adapter denominator not closed & Stop population rates; report registry and unresolved scope only \\
Decode & ABI/topic coverage incomplete & Exclude affected event family; retain raw envelope and anomaly \\
Linkage & Request cannot be exactly linked & Stop generation-path inference for affected rows; no temporal fallback \\
Clock & Required clock ungraded or unjustifiably point-valued & Use interval language or remove duration claim \\
Rule & Contemporaneous rule state unavailable & Remove rule-timing or early/late interpretation \\
Source & Evidence publication unavailable & Retain mechanism intervals; remove truth/source latency claim \\
Access & Historical access state unavailable & Report address concentration without access-regime comparison \\
Downstream finality & Protocol payout/redemption is outside the Part I estimand set & Report the claim in the companion Part II; do not infer it from Part I \\
Rebuild & Reported analytical outputs do not reproduce & No result based on those outputs is reported \\
\bottomrule
\end{tabularx}
\end{table}

\subsection{Claim-strength labels}

Every result is one of:

\begin{description}[style=nextline,leftmargin=3.0cm]
  \item[Population fact] Deterministic count or distribution on the frozen population.
  \item[Descriptive association] Estimated relationship without causal interpretation.
  \item[Interval observation] Event order or duration bounded but not point-observed.
  \item[Formal implication] Algebraic or state-representation result under explicit assumptions.
  \item[Design implication] Constraint on a later mechanism, not a production-safety claim.
\end{description}

\section{Source, Contract, and Claim Audit Protocol}
\label{app:sourceaudit}

Mechanism documentation, deployed addresses, verified source, dependency configuration, and historical invocation can change independently.  Claims therefore follow an explicit source hierarchy.

\subsection{Source hierarchy}

For historical mechanism claims, evidence is ranked as follows:

\begin{enumerate}
  \item on-chain bytecode, logs, storage, traces, and verified state at the relevant block;
  \item verified source matched to deployed bytecode and pinned to an immutable hash;
  \item signed governance proposal, deployment record, or official versioned specification;
  \item frozen official documentation with retrieval time and content hash;
  \item mutable current documentation;
  \item secondary discussion or reporting.
\end{enumerate}

Lower-ranked evidence can motivate a search but cannot override higher-ranked historical evidence.  Conflicts are reported rather than resolved rhetorically.

\subsection{Claim predicates}

Every mechanism statement receives one predicate:

\begin{description}[style=nextline,leftmargin=2.5cm]
  \item[DOCUMENTED] Official documentation currently states the behavior.
  \item[CAPABLE] Bytecode-matched source permits the behavior.
  \item[CONFIGURED] Historical state shows the relevant parameter or access control.
  \item[OBSERVED] A transaction, event, trace, or state transition shows that the behavior occurred.
  \item[INFERRED] A transparent inference from established facts; assumptions are stated.
\end{description}

``The adapter can manually resolve'' is CAPABLE.  ``Question $q$ was manually resolved'' requires OBSERVED evidence.  Current access configuration is never projected backward without historical state.

\subsection{Evidence freeze}

A reproducibility-grade evidence freeze includes:

\begin{itemize}
  \item deployed-address and proxy/implementation registry;
  \item runtime bytecode and source hashes;
  \item adapter dependency map;
  \item ABI/event-topic registry;
  \item source-claim matrix;
  \item initialized-question-weighted coverage audit;
  \item unresolved-items register with claim consequences;
  \item commands, software environment and revision, input/output hashes, row counts, and documented interventions.
\end{itemize}

A repository's current default branch is not presumed to represent an older deployed adapter.

\subsection{Historical-state discipline}

The following are not back-projected from current values:

\begin{itemize}
  \item proposer whitelist membership and enforcement;
  \item administrative roles;
  \item bond and liveness parameters;
  \item adapter-to-oracle dependencies;
  \item rule and clarification state;
  \item collateral and redemption routes.
\end{itemize}

When exact historical state is unavailable, the record remains interval-observed or unresolved and the corresponding claim is narrowed.

\section{Proofs and Formal Audit}
\label{app:proofs}

\begin{proof}[Proof of \Cref{prop:request_history}]
Let history \(H_0\) end with a proposal on the first active request and let \(H_1\) end with a
proposal on a successor request created after the question's first dispute and reset.  A coarse
representation omitting request generation gives the same current value:
\[
  \widetilde S_q(t;H_0)=\widetilde S_q(t;H_1)=\mathrm{Proposed}.
\]
Apply an otherwise identical valid dispute to the active request in each history.  Under the
branch assumed in the proposition, the dispute following \(H_0\) permits a successor-request
transition, whereas the dispute following \(H_1\) exposes the later-adjudication branch.
The feasible next-state set differs despite the same coarse state.  Hence the transition law
depends on omitted history, and active request generation or an equivalent sufficient variable is
required.
\end{proof}

\begin{proof}[Proof of \Cref{prop:scalar_timestamp}]
Let paths \(P\) and \(P'\) share \((Y_q,\tadapter)\).  Any function of that scalar pair assigns
them the same value.  From \eqref{eq:path_functional},
\[
  Z_q(P)-Z_q(P')
  =
  \sum_{j\in\mathcal J}h_j
  \left[\occupancy_{qj}(P)-\occupancy_{qj}(P')\right].
\]
By assumption, the occupation-vector difference is not orthogonal to \(\mathbf h\), so the
right-hand side is non-zero.  Therefore the path-dependent functional cannot be recovered from
the common scalar terminal pair.
\end{proof}

\begin{lemma}[Transport invariance of derived state]
\label{lem:transport_invariance}
Fix a contract/topic registry, inclusive block interval, canonical event identity, decoder, and
deterministic state builder.  If two providers return the same canonical event identities and
semantic payloads over the interval, the derived state tables are identical up to non-semantic
file metadata.
\end{lemma}

\begin{proof}
The state builder consumes the canonically ordered event set and fixed registries.  Equal
identities and semantic payloads imply equal ordered inputs.  Every state transition and derived
relation agrees by induction over that order.  Provider page boundaries, JSON formatting,
retrieval time, and file metadata do not enter the state function.
\end{proof}

\subsection{Formal-claim boundary}

The propositions are representation and identification results.  They do not prove that a
particular risk engine is optimal, that all historical questions used every source-code branch,
or that the missing external clocks can be recovered from on-chain timing.  Numerical examples
and figures are explanatory rather than proof substitutes.

\section{Reproducibility and Evidence Lineage}
\label{app:reproducibility}

The companion materials are evidence-graded and versioned.  Transport is treated as a means of
reconstructing immutable event relations rather than as the relations themselves.  Population,
transport, analysis, and disclosure roles remain distinct.

\subsection{Analytical table families}

The reported evidence includes adapter questions and events, Oracle request generations and
lifecycle events, question--request relations and path summaries, versioned rules and
clarifications, exact metadata matches and attrition, and the Part~I estimand and gate tables.
Protocol payout and redemption relations use the separately frozen downstream dataset considered
in Part~II.

\subsection{Identity and temporal fields}

Chain-event identity retains chain, contract, transaction hash, and log index.  Repeated source
observations do not create additional canonical events; they remain provenance edges.  Conflicting
payloads under one canonical identity are retained as conflicts.  Question and request identifiers
remain separate because a reset can create a successor request generation without creating a new
economic question.

Each timestamp records its clock type, point or interval representation, observation grade,
source identity, and block/transaction/log order where applicable.  Unmeasured external clocks
remain explicit unavailable coordinates rather than being filled with request or settlement time.

\subsection{Coverage and provenance}

Every target-event range records chain ID, contract, topic, inclusive block range, split lineage,
provider, response status, row count, and canonical-event digests.  A range is terminal only when
its response is successful and below the applicable cap, or when exact single-block pagination is
complete.  A broad first-page response does not prove coverage.

The durable event layer preserves complete topics and data, exact identity, decoder revision,
semantic hashes, range provenance, and deterministic audit samples.  Missing or malformed
load-bearing identity prevents use of the affected event in the corresponding claim.  Auxiliary
metadata remains explicitly missing when unavailable, and semantic conflicts are not resolved by
latest-file or latest-provider heuristics.

\subsection{Public reproducibility boundary}

The public companion supplies aggregate and privacy-reviewed material needed to assess the
reported claims.  It does not provide ranked pseudonymous addresses when aggregates suffice.
The accompanying documentation describes the relevant schema, methodology, limitations,
processing lineage, descriptive statistics, and citation metadata.  Later observation windows are
treated as new versions rather than silent extensions of the fixed snapshot.

\section{Version-Locked Hypotheses and Disposition}
\label{app:hypotheses}

The integrated analysis fixed nine directional hypotheses before the final accepted result
layers were inspected.  The split into Parts I--II and Papers~6 and~9 does not permit
those hypotheses to be rewritten after inspection.  \Cref{tab:hypothesis_disposition} records
their status.

\begin{table}[htbp]
\centering
\caption{Disposition of the version-locked hypotheses after the paper split.}
\label{tab:hypothesis_disposition}
\small
\begin{tabularx}{\textwidth}{@{}L{3.0cm}L{3.2cm}Y@{}}
\toprule
Hypothesis family & Status in Part I & Consequence \\
\midrule
Upper-tail finality duration &
Partially measured; original decision-ready endpoint unavailable &
Technical request/reset/settlement intervals are reported; the source-conditioned hypothesis is
not evaluated as stated. \\
Partial convergence at first terminal proposal &
Price and liquidity layer not evaluated in Part I &
No convergence effect is promoted in Part I. \\
Liquidity during the post-evidence finality gap &
Unmeasured external clock and price/liquidity layer &
Outside the Part I analytical scope; no substitute clock is used. \\
Clarification-path association &
Timing component measured; path-association claim pending &
Exact clarification ordering is reported without causal or incidence language beyond the
accepted tables. \\
Request-level proposal-access trade-off &
Access history partial observed-only &
Address concentration is reported; access-regime effects remain unestimated. \\
Finality wedge and duration &
Price and liquidity layer not evaluated in Part I &
Formal estimand retained; no numerical claim. \\
Bond-scaling mismatch &
Outside the Part I analytical scope &
Not a Part I result. \\
State-triggered margin and conversion &
Outside the Part I analytical scope &
Not a Part I result. \\
Administrative-route rarity and materiality &
Rarity observed; value materiality blocked &
Observed administrative counts are reported; economic materiality is not claimed. \\
\bottomrule
\end{tabularx}
\end{table}

The technical-initialization placebo remains a validity control rather than a substantive tenth
hypothesis.  A convergence response around initialization comparable to proposal/finality
transitions would require an identifier and leakage audit before supporting proposal-effect
language.

\section*{Data and Code Availability}
The article reports aggregate event, timing, and governance results; raw address-level and restricted reconstruction materials are not distributed.

\section*{Generative AI Disclosure}
OpenAI ChatGPT and Codex were used for editorial and technical assistance during manuscript preparation. The author made all substantive research decisions, reviewed the final manuscript, and assumes full responsibility for its contents.

\section*{Funding}
This research received no external funding.

\section*{Competing Interests}
The author is affiliated with the Research Department of Devnull FZCO and leads the ForesightFlow research programme. No external sponsor influenced the research design, analysis, interpretation, or decision to publish. The article does not evaluate a commercial product or make investment recommendations.

\section*{Acknowledgments}
This work was completed within the ForesightFlow research programme. All remaining errors are the author's responsibility.

\printbibliography[heading=bibintoc,title={References}]

@article{aalen_johansen_1978,
  author  = {Aalen, Odd O. and Johansen, Soren},
  title   = {An Empirical Transition Matrix for Non-Homogeneous Markov Chains Based on Censored Observations},
  journal = {Scandinavian Journal of Statistics},
  volume  = {5},
  number  = {3},
  pages   = {141--150},
  year    = {1978}
}

@unpublished{adler_astraea_2018,
  author        = {Adler, John and Berryhill, Ryan and Veneris, Andreas and Poulos, Zissis and Veira, Neil and Kastania, Anastasia},
  title         = {Astraea: A Decentralized Blockchain Oracle},
  year          = {2018},
  eprint        = {1808.00528},
  archivePrefix = {arXiv},
  doi           = {10.48550/arXiv.1808.00528}
}

@book{andersen_borgan_gill_keiding_1993,
  author    = {Andersen, Per Kragh and Borgan, Ornulf and Gill, Richard D. and Keiding, Niels},
  title     = {Statistical Models Based on Counting Processes},
  publisher = {Springer},
  address   = {New York},
  year      = {1993}
}

@report{cpss_iosco_2012,
  author      = {{Committee on Payment and Settlement Systems and Technical Committee of IOSCO}},
  title       = {Principles for Financial Market Infrastructures},
  institution = {Bank for International Settlements and International Organization of Securities Commissions},
  type        = {International standard},
  year        = {2012},
  url         = {https://www.bis.org/cpmi/publ/d101a.pdf},
  urldate     = {2026-08-13}
}

@inproceedings{eskandari_oracles_2021,
  author    = {Eskandari, Shayan and Salehi, Mehdi and Gu, Wanyun Catherine and Clark, Jeremy},
  title     = {{SoK}: Oracles from the Ground Truth to Market Manipulation},
  booktitle = {Proceedings of the 3rd ACM Conference on Advances in Financial Technologies},
  year      = {2021},
  doi       = {10.1145/3479722.3480994},
  eprint    = {2106.00667},
  archivePrefix = {arXiv}
}

@article{fine_gray_1999,
  author  = {Fine, Jason P. and Gray, Robert J.},
  title   = {A Proportional Hazards Model for the Subdistribution of a Competing Risk},
  journal = {Journal of the American Statistical Association},
  volume  = {94},
  number  = {446},
  pages   = {496--509},
  year    = {1999},
  doi     = {10.1080/01621459.1999.10474144}
}

@article{hanson_2003,
  author  = {Hanson, Robin},
  title   = {Combinatorial Information Market Design},
  journal = {Information Systems Frontiers},
  volume  = {5},
  number  = {1},
  pages   = {107--119},
  year    = {2003},
  doi     = {10.1023/A:1022058209073}
}

@article{kaplan_meier_1958,
  author  = {Kaplan, Edward L. and Meier, Paul},
  title   = {Nonparametric Estimation from Incomplete Observations},
  journal = {Journal of the American Statistical Association},
  volume  = {53},
  number  = {282},
  pages   = {457--481},
  year    = {1958},
  doi     = {10.1080/01621459.1958.10501452}
}

@unpublished{kota_ai_oracles_2026,
  author        = {Kota, Tarun},
  title         = {Design and Evaluation of Multi-Agent AI Oracle Systems for Prediction Market Resolution},
  year          = {2026},
  eprint        = {2605.30802},
  archivePrefix = {arXiv},
  doi           = {10.48550/arXiv.2605.30802}
}

@article{mackinlay_1997,
  author  = {MacKinlay, A. Craig},
  title   = {Event Studies in Economics and Finance},
  journal = {Journal of Economic Literature},
  volume  = {35},
  number  = {1},
  pages   = {13--39},
  year    = {1997}
}

@article{manski_2006,
  author  = {Manski, Charles F.},
  title   = {Interpreting the Predictions of Prediction Markets},
  journal = {Economics Letters},
  volume  = {91},
  number  = {3},
  pages   = {425--429},
  year    = {2006},
  doi     = {10.1016/j.econlet.2005.12.004}
}

@inproceedings{muehlberger_oracle_patterns_2020,
  author    = {M{\"u}hlberger, Roman and Bachhofner, Stefan and Castell{\'o} Ferrer, Eduardo and Di Ciccio, Claudio and Weber, Ingo and W{\"o}hrer, Maximilian and Zdun, Uwe},
  title     = {Foundational Oracle Patterns: Connecting Blockchain to the Off-Chain World},
  booktitle = {Business Process Management: Blockchain and Robotic Process Automation Forum},
  pages     = {35--51},
  year      = {2020},
  publisher = {Springer},
  doi       = {10.1007/978-3-030-58779-6_3},
  eprint    = {2007.14946},
  archivePrefix = {arXiv}
}

@online{polymarket_adapter_auth_2026,
  author  = {{Polymarket}},
  title   = {UmaCtfAdapter Auth Solidity Source},
  year    = {2026},
  url     = {https://github.com/Polymarket/uma-ctf-adapter/blob/8b76cc9e0d46c6f7450a0adb0ddc0f5b0568c9cc/src/mixins/Auth.sol},
  urldate = {2026-07-12}
}

@online{polymarket_bulletin_board_2026,
  author  = {{Polymarket}},
  title   = {BulletinBoard Solidity Source},
  year    = {2026},
  url     = {https://github.com/Polymarket/uma-ctf-adapter/blob/8b76cc9e0d46c6f7450a0adb0ddc0f5b0568c9cc/src/mixins/BulletinBoard.sol},
  urldate = {2026-07-12}
}

@online{polymarket_perps_faq_2026,
  author  = {{Polymarket}},
  title   = {Polymarket Perps FAQ},
  year    = {2026},
  url     = {https://docs.polymarket.com/perps/faq},
  urldate = {2026-07-12}
}

@online{polymarket_perps_markets_2026,
  author  = {{Polymarket}},
  title   = {Polymarket Perps: Markets},
  year    = {2026},
  url     = {https://docs.polymarket.com/perps/learn-about-trading/markets},
  urldate = {2026-07-12}
}

@online{polymarket_resolution_docs_2026,
  author  = {{Polymarket}},
  title   = {Resolution},
  year    = {2026},
  url     = {https://docs.polymarket.com/concepts/resolution},
  urldate = {2026-07-12}
}

@online{polymarket_resolution_subgraph_2026,
  author  = {{Polymarket}},
  title   = {Resolution Subgraph: Polygon Adapter Data Sources and Resolution Mappings},
  year    = {2026},
  url     = {https://github.com/Polymarket/resolution-subgraph},
  urldate = {2026-07-13},
  note    = {Operational source freeze; adapter registry and event-handler configuration pinned at commit 75d1818547862a5bd3477ed2e6b16f693d42dab6}
}

@online{polymarket_uma_adapter_contract_2026,
  author  = {{Polymarket}},
  title   = {UmaCtfAdapter Solidity Source},
  year    = {2026},
  url     = {https://github.com/Polymarket/uma-ctf-adapter/blob/8b76cc9e0d46c6f7450a0adb0ddc0f5b0568c9cc/src/UmaCtfAdapter.sol},
  urldate = {2026-07-12}
}

@online{polymarket_uma_adapter_repo_2026,
  author  = {{Polymarket}},
  title   = {UMA CTF Adapter},
  year    = {2026},
  url     = {https://github.com/Polymarket/uma-ctf-adapter},
  urldate = {2026-07-12}
}

@article{turnbull_1976,
  author  = {Turnbull, Bruce W.},
  title   = {The Empirical Distribution Function with Arbitrarily Grouped, Censored and Truncated Data},
  journal = {Journal of the Royal Statistical Society: Series B},
  volume  = {38},
  number  = {3},
  pages   = {290--295},
  year    = {1976},
  doi     = {10.1111/j.2517-6161.1976.tb01597.x}
}

@online{uma_moov2_2026,
  author  = {{UMA}},
  title   = {ManagedOptimisticOracleV2},
  year    = {2026},
  url     = {https://docs.uma.xyz/developers/managedoptimisticoraclev2},
  urldate = {2026-07-12}
}

@online{uma_moov2_programmatic_2026,
  author  = {{UMA}},
  title   = {ManagedOptimisticOracleV2: Proposing Programmatically},
  year    = {2026},
  url     = {https://docs.uma.xyz/developers/managedoptimisticoraclev2/proposing-programmatically},
  urldate = {2026-07-12}
}

@online{uma_moov2_whitelist_2026,
  author  = {{UMA}},
  title   = {ManagedOptimisticOracleV2: Default Proposer Whitelist},
  year    = {2026},
  url     = {https://docs.uma.xyz/developers/managedoptimisticoraclev2/default-proposer-whitelist},
  urldate = {2026-07-12}
}

@online{uma_oracle_overview_2026,
  author  = {{UMA}},
  title   = {How Does UMA's Oracle Work?},
  year    = {2026},
  url     = {https://docs.uma.xyz/protocol-overview/how-does-umas-oracle-work},
  urldate = {2026-07-12}
}

@online{uma_polymarket_verification_2026,
  author  = {{UMA}},
  title   = {Verification Guide: Polymarket and YES\_OR\_NO\_QUERY},
  year    = {2026},
  url     = {https://docs.uma.xyz/verification-guide/yes_or_no},
  urldate = {2026-07-12}
}

@article{wolfers_zitzewitz_2004,
  author  = {Wolfers, Justin and Zitzewitz, Eric},
  title   = {Prediction Markets},
  journal = {Journal of Economic Perspectives},
  volume  = {18},
  number  = {2},
  pages   = {107--126},
  year    = {2004},
  doi     = {10.1257/0895330041371321}
}

@misc{nechepurenko2026_p1,
  author        = {Nechepurenko, Maksym},
  title         = {Resolution-Aware Perpetual Futures on Binary Prediction Markets: Failure Modes and Mechanical Stress Tests Using Polymarket Data},
  year          = {2026},
  url           = {https://ssrn.com/abstract=6748278},
  note          = {SSRN abstract 6748278; arXiv:2605.10400},
  eprint        = {2605.10400},
  archiveprefix = {arXiv},
  primaryclass  = {q-fin.TR},
  doi           = {10.2139/ssrn.6748278}
}

@misc{nechepurenko2026_p2,
  author        = {Nechepurenko, Maksym},
  title         = {A Taxonomy of Event-Linked Perpetual Futures: Design Axes, Failure Modes, and Empirical Evaluability},
  year          = {2026},
  url           = {https://ssrn.com/abstract=6748298},
  note          = {SSRN abstract 6748298; arXiv:2605.10428},
  eprint        = {2605.10428},
  archiveprefix = {arXiv},
  primaryclass  = {q-fin.TR},
  doi           = {10.2139/ssrn.6748298}
}

@misc{nechepurenko2026_p3,
  author        = {Nechepurenko, Maksym},
  title         = {Manipulation, Informed Trading, and Regulation in Leveraged Event-Linked Markets},
  year          = {2026},
  url           = {https://ssrn.com/abstract=6748318},
  note          = {SSRN abstract 6748318; arXiv:2605.10486},
  eprint        = {2605.10486},
  archiveprefix = {arXiv},
  primaryclass  = {q-fin.TR},
  doi           = {10.2139/ssrn.6748318}
}

@misc{nechepurenko2026_p4,
  author        = {Nechepurenko, Maksym},
  title         = {Fill-Side Non-Retail Trading on Polymarket: Identification Limits, Behavioral Tiers, and Microstructure Signatures},
  year          = {2026},
  url           = {https://ssrn.com/abstract=6751284},
  note          = {SSRN abstract 6751284; arXiv title differs: Fill-Side Behavioral Concentration on Polymarket: Identification Limits under Record-Level Attribution.},
  eprint        = {2605.11640},
  archiveprefix = {arXiv},
  primaryclass  = {q-fin.TR},
  doi           = {10.2139/ssrn.6751284}
}
\end{document}